\documentclass[11pt]{article}
\usepackage{blindtext}
\usepackage[usenames,dvipsnames,svgnames,table]{xcolor}
\usepackage{enumitem}
\usepackage{caption}
\usepackage{placeins}
\usepackage{graphicx}

\usepackage{fancybox}
\usepackage{hyperref}

\usepackage{amsmath,amsfonts,amsthm,amssymb,xcolor}
\usepackage{bm} 

\usepackage{nicefrac}

\usepackage{float}

\newtheorem{theorem}{Theorem}[section]

\newtheorem{corollary}[theorem]{Corollary}
\newtheorem{lemma}[theorem]{Lemma}

\newtheorem{claim}[theorem]{Claim}

\newtheorem{fact}[theorem]{Fact}

\newtheorem*{theorem*}{Theorem}
\newtheorem*{corollary*}{Corollary}
\newtheorem*{conjecture*}{Conjecture}
\newtheorem*{lemma*}{Lemma}
\newtheorem*{thm*}{Theorem}
\newtheorem*{prop*}{Proposition}
\newtheorem*{obs*}{Observation}

\newtheorem*{rec*}{Recommendation}

\theoremstyle{definition}
\newtheorem{definition}[theorem]{Definition}

\newtheorem*{remark*}{Remark}
\newtheorem*{definition*}{Definition}

\newenvironment{fminipage}%
  {\begin{Sbox}\begin{minipage}}%
  {\end{minipage}\end{Sbox}\fbox{\TheSbox}}

\def\defeq{\stackrel{\mathrm{def}}{=}}
\def\setof#1{\left\{#1  \right\}}
\def\sizeof#1{\left|#1  \right|}

\def\union{\cup}
\def\intersect{\cap}

\def\abs#1{\left|#1  \right|}

\def\norm#1{\left\| #1 \right\|}

\def\calA{\mathcal{A}}
\def\calC{\mathcal{C}}
\def\calD{\mathcal{D}}

\def\calG{\mathcal{G}}

\def\calK{\mathcal{K}}
\def\calL{\mathcal{L}}
\def\calS{\mathcal{S}}
\def\calT{\mathcal{T}}
\def\calM{\mathcal{M}}

\newcommand\PPi{\boldsymbol{\Pi}}

\newcommand\aalpha{\boldsymbol{\alpha}}

\newcommand{\assign}{\leftarrow}

\def\aa{\pmb{\mathit{a}}}
\newcommand\bb{\boldsymbol{\mathit{b}}}
\newcommand\cc{\boldsymbol{\mathit{c}}}
\newcommand\dd{\boldsymbol{\mathit{d}}}

\newcommand\ff{\boldsymbol{\mathit{f}}}
\renewcommand\gg{\boldsymbol{\mathit{g}}}

\newcommand\pp{\boldsymbol{\mathit{p}}}
\newcommand\qq{\boldsymbol{\mathit{q}}}

\newcommand\uu{\boldsymbol{\mathit{u}}}

\newcommand\yy{\boldsymbol{\mathit{y}}}
\newcommand\zz{\boldsymbol{\mathit{z}}}
\newcommand\xx{\boldsymbol{\mathit{x}}}

\newcommand\veczero{\boldsymbol{0}}
\newcommand\vecone{\boldsymbol{1}}
\newcommand{\zero}{\mathbf{0}}

\newcommand\bz{\mathbb Z}

\newcommand\rea{\mathbb R}

\renewcommand\AA{\boldsymbol{\mathit{A}}}
\newcommand\BB{\boldsymbol{\mathit{B}}}
\newcommand\CC{\boldsymbol{\mathit{C}}}

\newcommand\GG{\boldsymbol{\mathit{G}}}
\newcommand\HH{\boldsymbol{{H}}}
\newcommand\II{\boldsymbol{\mathit{I}}}

\newcommand\MM{\boldsymbol{\mathit{M}}}

\newcommand\WW{\boldsymbol{\mathit{W}}}

\newcommand\BBtil{\boldsymbol{\mathit{\widetilde{B}}}}

\newcommand\WWtil{\boldsymbol{\mathit{\widetilde{W}}}}

\newcommand\AAhat{\boldsymbol{\widehat{\mathit{A}}}}

\newcommand\BBhat{\boldsymbol{\widehat{\mathit{B}}}}

\newcommand\fftil{\boldsymbol{\tilde{\mathit{f}}}}
\newcommand\xxtil{\boldsymbol{\tilde{\mathit{x}}}}

\newcommand\Otil{\widetilde{O}}

\newcommand\R{\mathbb{R}}

\newcommand{\trp}{\top}

\DeclareMathOperator{\nnz}{nnz}

\DeclareMathOperator*{\argmin}{arg\,min}

\DeclareMathOperator*{\diag}{diag}
\DeclareMathOperator*{\Span}{span}
\DeclareMathOperator*{\conv}{conv}
\DeclareMathOperator*{\Null}{null}

\DeclareMathOperator*{\im}{im}

\DeclareMathOperator*{\poly}{poly}

\def\eps{\epsilon}

\def\lsa{\textsc{lea}}
\def\ls{\textsc{le}}
\def\slec{\textsc{sle}-complete}

\def\one{{\bf 1}}
\def\epsda{\epsilon^{DA}}
\def\epsb{\epsilon^{B_2}}
\def\calP{\mathcal{P}}
\def\da{\calD\calA}

\def\b2{\mathcal{B}_2}

\def\aalpha{\boldsymbol{\alpha}}
\def\ttil{\tilde{t}}
\def\mtil{\tilde{m}}

\def\triangle{\boldsymbol{\Delta}}

\usepackage{fullpage}

\usepackage[letterpaper,
left=1in,right=1in,top=1in,bottom=1in]{geometry}

\usepackage[
  backend=biber,
  backref=true,
  backrefstyle=none,
  date=year,
  doi=false,
  giveninits=true,
  hyperref=true,
  maxbibnames=10,
  sortcites=false,
  style=alphabetic,
  url=false, 
]{biblatex}
\usepackage{amsmath}
\usepackage{amssymb}
\usepackage{amsthm}
\usepackage{graphicx}
\usepackage{subfigure}
\usepackage{float}
\usepackage{color}
\usepackage{pifont}
\usepackage[ruled, linesnumbered]{algorithm2e}
\usepackage{booktabs}
\usepackage{multirow}
\usepackage{threeparttable}
\usepackage{arydshln}
\usepackage{multirow}
\usepackage{threeparttable}
\usepackage[toc,page]{appendix}
\usepackage{bbm}
\usepackage{mathrsfs} 
\usepackage{tablefootnote}
\usepackage{blkarray}
\usepackage{nicematrix}

\newcommand{\todo}[1]{{\bf \color{red} TODO: #1}}
\newcommand{\todolow}[1]{{\bf \color{o range} TODOLOW: #1}}

\newcommand{\rasmus}[1]{{\bf \color{orange} Rasmus: #1}}
\newcommand{\ming}[1]{{\bf \color{blue} [Ming: #1]}}
\newcommand{\peng}[1]{{\bf \color{orange}[Peng: #1]}}

 \renewcommand\todo[1]{}
 \renewcommand{\todolow}[1]{}
 \renewcommand{\rasmus}[1]{}
 \renewcommand{\ming}[1]{}
 \renewcommand{\peng}[1]{}

\begin{document}

\title{Hardness Results for Laplacians of Simplicial Complexes via
  Sparse-Linear Equation Complete Gadgets}

\author{
	Ming Ding\\ 
	\texttt{ming.ding@inf.ethz.ch}\\
	Department of Computer Science\\
	ETH Zurich
	\and
	Rasmus Kyng\thanks{The research leading to these results has received funding from grant no. 200021 204787 of the Swiss National Science Foundation.}\\ 
	\texttt{kyng@inf.ethz.ch}\\
	Department of Computer Science\\
	ETH Zurich
	\and
	Maximilian Probst Gutenberg\footnotemark[1]\\ 
	\texttt{maximilian.probst@inf.ethz.ch}\\
	Department of Computer Science\\
	ETH Zurich
	\and
	Peng Zhang\\
	\texttt{pz149@rutgers.edu}\\
	Department of Computer Science\\
	Rutgers University
}

\date{} 

\clearpage\maketitle
\thispagestyle{empty}

\begin{abstract}
  We study linear equations in combinatorial Laplacians of $k$-dimensional simplicial complexes ($k$-complexes), a natural generalization of graph Laplacians. 
  Combinatorial Laplacians play a crucial role in homology and are a central tool in topology. Beyond this, they have various applications in data analysis and physical modeling problems. 
  It is known that nearly-linear time solvers exist for graph Laplacians. However, nearly-linear time solvers for combinatorial Laplacians are only known for restricted classes of complexes.

  This paper shows that linear equations in combinatorial Laplacians of 2-complexes are as hard to solve as general linear equations. 
  More precisely, for any constant $c \geq 1$, if we can solve linear equations in combinatorial Laplacians of 2-complexes up to high accuracy in time $\tilde{O}((\# \text{ of nonzero coefficients})^c)$, then we can solve general linear equations with polynomially bounded integer coefficients and condition numbers up to high accuracy in time $\tilde{O}((\# \text{ of nonzero coefficients})^c)$. We prove this by a nearly-linear time reduction from general linear equations to combinatorial Laplacians of 2-complexes.
  Our reduction preserves the sparsity of the problem instances up to poly-logarithmic factors.

\end{abstract}

\newpage
\pagenumbering{gobble}

\sloppy

{\footnotesize
	\tableofcontents
}

\newpage

\pagenumbering{arabic}


\clearpage

\section{Introduction}

\subsection{Simplicial Complexes, Homology, and Combinatorial Laplacians}
\label{sec:background}
We study linear equations whose coefficient matrices are combinatorial Laplacians of $k$-dimensional abstract simplicial complexes ($k$-complexes), which generalize the well-studied graph Laplacians. An abstract simplicial complex $\calK$ is a family of sets, known as simplices, closed under taking subsets, i.e., every subset of a set in $\calK$ is also in $\calK$. The dimension of $\calK$ is the largest size of the simplices in $\calK$ minus 1. A geometric notion of abstract simplicial complexes is simplicial complexes, under which a $k$-simplex is the convex hull of $k+1$ vertices (for example, 0,1,2-simplexes are vertices, edges, and triangles, respectively). In particular, complexes in 1 dimension are graphs; combinatorial Laplacians in 1-complexes are graph Laplacians.

Nearly-linear time solvers exist for linear equations in graph Laplacians~\cite{ST14,KMP10,KMP11,PS14,CKMPPRX14,KS16,JS21}, and some natural generalizations such as connection Laplacians~\cite{KLPSS16} and directed Laplacians~\cite{CKPPRSV17,CKKPPRS18}. However, nearly-linear time solvers for linear equations in combinatorial Laplacians are only known for very restricted classes of 2-complexes~\cite{CFMNPW14, BMNW22}. We ask whether one can extend these nearly-linear solvers to general combinatorial Laplacians.

Combinatorial Laplacians are defined via boundary operators of the chain spaces of an oriented complex. Given an oriented simplicial complex $\calK$, a $k$-chain is a (signed) weighted sum of the $k$-simplices in $\calK$. The boundary operator $\partial_k$ is a linear map from the $k$-chain space to the $(k-1)$-chain space; in particular, it maps a $k$-simplex to a signed sum of its boundary $(k-1)$-simplices, where the signs are determined by the orientations of the $k$-simplex and its boundary $(k-1)$-simplices. For example, $\partial_1$ is the oriented vertex-edge incidence matrix. The combinatorial Laplacian $\calL_k$ is defined to be $\partial_{k+1} \partial_{k+1}^\top + \partial_{k}^\top \partial_{k}$. In particular, $\calL_0 = \partial_1 \partial_1^\top$ is the graph Laplacian.

Combinatorial Laplacians play an important role in both pure
mathematics and applied areas. These matrices originate in the study
of discrete Hodge decomposition~\cite{eckmann44}: The kernel of
$\calL_k$ is isomorphic to the $k$th homology space of $\calK$. The
properties of combinatorial Laplacians have been studied in many subsequent
works \cite{friedman98, DW02, DKM09, DKM15, mn21}. A central problem
in homology theory is evaluating the Betti number of the $k$th
homology space, which equals the rank of $\calL_k$. In the case of
homology over the reals, computing the rank of $\calL_k$ can be
reduced to solving a poly-logarithmic number of linear equations in
$\calL_k$ \cite{BV21}. Computation of Betti numbers over the reals is
a key step in numerous problems in applied topology, computational
topology, and topological data analysis \cite{Z05, G08, C09, EH10,
  CSGO16}. In addition, combinatorial Laplacians have applications in
statistical ranking \cite{JLYT11,XHJYLY12}, graphics and image
processing \cite{MMOC11,TLHD03},
electromagnetism and fluids mechanics \cite{DKT08}, data representations \cite{CMZ18}, cryo-electron microscopy \cite{KL17}, biology \cite{SBHLJ20}. We refer to the readers to \cite{lim20} for an accessible survey.

The reader may be puzzled that despite a vast literature on
combinatorial Laplacians and their central role in topology, little is
known about solving linear equations in these matrices except in very
restricted cases \cite{CFMNPW14,BMNW22}. In this paper, we show that approximately solving linear equations in general combinatorial Laplacians
is as hard as approximately solving general linear equations over the
reals, which explains the lack of special-purpose solvers for this
class of equations. More precisely, if one can solve linear equations
in combinatorial Laplacians of general $2$-complexes to high accuracy
in time $\tilde{O}((\# \text{ of nonzero
  coefficients})^c)$\footnote{$\tilde{O}$ hides poly-logarithmic
  factors in following parameters of the input: ratio of maximum and
  minimum non-zero singular values, the maximum ratio of non-zero entries
  (in absolute value), and the inverse of the accuracy parameter.} 
for some constant $c \ge 1$, then one can solve general linear equations with polynomially bounded integer coefficients and condition numbers
up to high accuracy in time $\tilde{O}((\# \text{ of nonzero
  coefficients})^c)$. A recent breakthrough shows that general linear
equations can be solved up to high accuracy in time $\Otil((\# \text{
  of nonzero coefficients})^{2.27159})$~\cite{PV21,nie21}, which for
sparse linear equations is asymptotically faster than the
long-standing runtime barrier of fast matrix multiplication \cite{S69},
which currently achieves a running time of $\tilde{O}(n^{2.3728596})$~\cite{AW21}.
Understanding the optimal value of $c$ is a major open problem in
numerical linear algebra.
Our result, viewed positively, shows that one can reduce the problem
of designing fast solvers for general linear equations 
to that for combinatorial Laplacians.

\subsection{Hardness Results Based on Linear Equations}

Kyng and Zhang~\cite{KZ17} initiated the study of hardness results for solving structured linear equations. 
They showed that solving linear equations in a slight generalization of graph Laplacians such as 2-commodity Laplacians, 2-dimensional truss stiffness matrices, and 2-total-variation matrices is as hard as solving general linear equations.

Suppose given an invertible matrix $\AA$ and a vector $\bb$ over the reals, 
we want to approximately solve linear equation $\AA \xx =
\bb$, 
i.e., find $\xxtil$ such that $\norm{\AA \xxtil - \bb}_2 \leq \epsilon
\norm{\bb}_2$ for some $\epsilon$.

\begin{definition*}[(Informal) Sparse-linear-equation
    completeness of matrix family $\mathcal{B}$.]
Consider a family of matrices $\mathcal{B}$, and suppose that for
any instance $(\AA,\bb,\epsilon)$ 
we can produce matrix $\BB \in \mathcal{B}$, vector $\cc$, and accuracy parameter $\delta$, such that if we can solve $\BB \yy = \cc$ up to error $\delta$, 
then we can produce $\xxtil$ that solves $\AA \xx = \bb$ to the
desired accuracy.

If, given  $(\AA,\bb,\epsilon)$,  we can compute $(\BB,\cc,\delta)$ in
$\tilde{O}(\nnz(\AA))$ time with $\nnz(\BB) = \tilde{O}(\nnz(\AA))$
then we say that the class $\mathcal{B}$ is \emph{sparse-linear-equation complete}.
\end{definition*}

In our preliminaries in Section~\ref{sect:preliminaries}, we state a
formal definition of sparse-linear equation completeness that also
extends to non-invertible matrices.

The reason for our use of the term ``completeness'' is that if a solver 
with runtime $\tilde{O}((\# \text{ of nonzero coefficients})^c)$ is known for the class $\mathcal{B}$,
then a solver with runtime $\tilde{O}((\# \text{ of nonzero
  coefficients})^c)$ exists for general matrices.
Such solvers are known for the classes of Laplacian Matrices, Directed Laplacian Matrices, Connection Laplacian Matrices, and several more classes, all with $c = 1$. Thus, if any of
these classes were sparse-linear-equation complete, we would
immediately get nearly-linear time solvers for general linear
equations.

In this language, Kyng and Zhang~\cite{KZ17} showed that 2-commodity
Laplacians, 2-dimensional truss stiffness matrices, and 2-total-variation matrices are all sparse-linear-equation complete.
 We note that \cite{KWZ20} considered a larger family of hardness assumptions based on linear equations, 
 which, among other things, can express weaker hardness statements based on weaker reductions.

\subsection{Our Contributions}

In the terminology established above, our main result can be stated
very succinctly:
\begin{theorem}[Informal First Main Theorem]
  \label{thm:informalCombLapHard}
  Linear equations in combinatorial Laplacians of 2-complexes are sparse-linear-equation complete.
\end{theorem}
In fact, we show this by showing an even simpler problem is
sparse-linear-equation complete, namely linear equations in the
boundary operator of a 2-complex, which is our second main result.
\begin{theorem}[Informal Second Main Theorem]
  Linear equations in the boundary operators $\partial_2$ of 2-complexes are sparse-linear-equation complete.
\end{theorem}
This result is formally stated in Theorem~\ref{thm:boundarySLEC}.
Below, in Section~\ref{sect:boundary_and_combinatorial}, we sketch how
our first main theorem above follows from our second main theorem. We
give a formal proof of this in Appendix~\ref{sect:appendix_lap2boundary}.

Our proof establishes the sparse-linear-equation completeness of
boundary operators of a 2-complex via a two-step reduction.
In our first reduction step, we show sparse-linear-equation completeness of a very
simple class of linear equations which we call \emph{difference-average
equations}: these are equations where every row either restricts the
difference of two variables:
$\xx(i) - \xx(j) = b$, or it sets one variable to be the average of
two others: $\xx(i) + \xx(j) = 2\xx(k)$.
This reduction was implicitly proved in~\cite{KZ17} as an intermediate
step.
In this paper, we make the reduction explicit, which may be of
independent interest, as this reduction class is likely to be a good
starting point for many other hardness reductions.
One can think of this step as analogous to showing that 3-SAT is
NP-complete: It gives us a simple starting point for proving the hardness of other problems.
The formal theorem statement appears in Theorem~\ref{thm:prelim_kz17}.
In our second reduction step, we reduce a given difference-average
equation problem to a linear equation in the boundary operator of a
2-complex.

Both the two steps preserve the number of nonzero coefficients in the linear equations
up to a logarithmic factor, and only blow up the coefficients and condition numbers polynomially.
The reductions are also robust to error in the sense that to solve the
original problem to high accuracy, it suffices to solve the
reduced problem to accuracy at most polynomially higher.
Finally, we can compose the two reductions to show that solving linear
equations in 2-complex boundary operators to high accuracy is as hard as solving
general linear equations with polynomially bounded integer
coefficients and condition numbers to high accuracy.

We give more details on both reductions below, but first we describe how to
show that solving linear equations in combinatorial Laplacians
$\calL_1$ is also sparse-linear-equation complete.

\subsubsection{Hardness for Combinatorial Laplacians From Hardness for Boundary Operators}
\label{sect:boundary_and_combinatorial}
Our main technical result, Theorem~\ref{thm:boundarySLEC}, shows that the class of linear equations in the
boundary operators of 2-complexes is sparse-linear-equation
complete.
But, as the following simple lemma shows, we can reduce the problem of
solving in a boundary operator $\partial_2$ to solving in the corresponding
combinatorial Laplacian $\calL_1$, and hence the latter problem must
be at least as hard.
This then immediately implies our first main result,
Theorem~\ref{thm:informalCombLapHard}.
The reduction is captured in the following lemma.
\begin{lemma}[Informal reduction from boundary operators to
  combinatorial Laplacians in 2-complexes]
  Suppose we can solve linear equations in combinatorial Laplacians of 2-complexes 
  to high accuracy in nearly-linear time. Then, we can 
  solve linear equations in boundary operators $\partial_2$ of 2-complexes 
  to high accuracy in nearly-linear time.
  \label{thm:compLap2Boundary_informal}
\end{lemma}
The proof is by standard arguments which we sketch below. In Appendix \ref{sect:appendix_lap2boundary}, 
we will formally state the theorem and provide a rigorous proof.
Suppose we have a high-accuracy solver 
for combinatorial Laplacians of 2-complexes.
Using this, we want to obtain a solver for linear equations in the
boundary operator $\partial_2$.
Note that when the equation $\partial_2 \ff= \dd$ is infeasible, we
measure the solution quality by $\norm{\partial_2 \ff -
  \PPi_{\partial_2} \dd}_2$ where $\PPi_{\partial_2} $ denotes the orthogonal
projection onto the image $\im(\partial_2)$.
The minimum over $\ff$ of the quantity $\norm{\partial_2 \ff -
  \PPi_{\partial_2} \dd}_2$ is zero, which
is obtained by setting $\ff =
\partial_2^{\dagger} \dd$ (where $\partial_2^{\dagger}$ is the
Moore-Penrose pseudo-inverse of $\partial_2$).
The equation $\partial_2 \ff= \dd$ is feasible exactly when
$\PPi_{\partial_2} \dd = \dd$.

A central and basic fact in the study of simplicial homology is that
$\im(\partial_1^{\top}) \intersect \im(\partial_2) =
\setof{\veczero}$.
This implies that $\PPi_{\partial_2} \partial_1^{\top} = \veczero$.
Now, suppose that $\xxtil$ approximately solves 
$\calL_1 \xx  = \dd$, i.e. $\calL_1 \xxtil  \approx \dd$.
We can rewrite this as $\partial_1^\top \partial_1 \xxtil +
\partial_2\partial_2^\top \xxtil   \approx \dd$.
Now, if we apply $\PPi_{\partial_2}$ on both sides,
$\PPi_{\partial_2} \dd \approx \PPi_{\partial_2}
\partial_2\partial_2^\top \xxtil  = \partial_2\partial_2^\top \xxtil$.
Thus if we set $\fftil = \partial_2^\top \xxtil$, then we have
$\PPi_{\partial_2} \dd \approx \partial_2 \fftil$, which matches our
notion of $\fftil$ approximately solving the (possibly infeasible) 
linear equation $\partial_2 \ff= \dd$.
This means that if we can approximately solve linear equations in 
$\calL_1$, we can solve linear
equations in $\partial_{2}$.
This way we can also argue that if we can solve linear
equations in $\partial_2\partial_2^\top$, then we can solve linear
equations in $\partial_2$.
Finally, one should note that $\nnz(\calL_1) = O(\nnz(\partial_2))$
and that using our definition of condition number (see
Section~\ref{sect:preliminaries}), both have polynomially related
condition number\footnote{This is because $\partial_1$ has polynomially
  bounded singular values.}.
This also means a high accuracy solve in one
can be converted to a high accuracy solve in the other.

\subsubsection{Linear Equations in $\partial_2 \partial_2^\top$}
In addition to the many applications discussed in
Section~\ref{sec:background}, the problem of solving linear equations
in $\partial_2 \partial_2^\top$
also arises when using Interior Point Methods to solve  a generalized max-flow problem in higher-dimensional
simplicial complexes as defined in~\cite{mn21}. 
We sketch how this inverse problem arises when using an Interior Point Method in Appendix \ref{sect:appendix_IPM}. 
By a similar argument as Lemma \ref{thm:compLap2Boundary_informal}, we can show that if we can solve 
linear equations in $\partial_2 \partial_2^\top$ to high accuracy in nearly-linear time, then we can solve 
linear equations in $\partial_2 $ to high accuracy in nearly-linear time.

\subsubsection{Sparse-Linear-Equation Completeness of Difference-Average Equations}

Our first reduction transforms general linear equations with polynomially bounded integer entries and condition numbers into difference-average equations. 
We first transform a general linear equation instance to a linear equation instance such that the coefficient matrix has row sum zero
and the sum of positive coefficients in each row is a power of 2, by introducing a constant number of more variables and equations.
Then, we transform each single equation to a set of difference-average equations 
by bit-wise pairing and replacing each pair of variables 
with a new variable via an average equation.

\subsubsection{Sparse-Linear-Equation Completeness of Boundary Operators of Simplicial Complexes}
Our second reduction transforms difference-average linear equations into linear equations in the boundary operators of 2-complexes.
Solving $\partial_2 \ff = \dd$ can be interpreted as computing a flow $\ff$ in the triangle space of a 2-complex subject to pre-specified edge demands $\ff$.

Our reduction 
is inspired by a reduction in~\cite{mn21} that proves 
NP-hardness of computing maximum \emph{integral} flows in 2-complexes 
via a reduction from graph 3-coloring problem. %
However, the correctness of their reduction heavily
relies on that the flow values in the 2-complex are 0-1 integers,
which does not apply in our setting.
 In addition, it is unclear how to encode linear equations as a graph coloring problem even if fractional colors are allowed.

We employ some basic building blocks used in \cite{mn21} including punctured spheres and tubes. 
However, we need to carefully arrange and orient the triangles in the 2-complex
to encode both the positive and negative coefficients in difference-average equations, 
and we need to express the averaging relations not covered by the previous work.

An important aspect of our contribution is that we carefully control the
number of non-zeros of the boundary operator matrix that we construct,
and we bound the condition number of this matrix and how error propagates from an approximate solution to the
boundary operator problem back to the original difference-average equations.
In order to do so, 
we develop explicit triangulation algorithms that specify the precise number of triangles
needed to triangulate each building block and allow a detailed error
and condition number analysis.

We remark that our constructed 2-complex does not admit an embedding into a sphere in
3 dimensions.
Recent work \cite{BMNW22} has shown that simplicial complexes with a
known embedding into $\R^3$ have non-trivial linear equation solvers,
but the full extent to which embeddability can lead to better solvers
remains an open question.

We analyze our construction in the Real RAM model. However, it
can be transferred to the fixed point arithmetic model with $(\log N)^{O(1)}$ bits per number, where $N$ is the size of the problem instance.

\subsection{Related Works}

\paragraph{Generalized Flows.}
One can generalize the notions of flows, demands vectors,
circulations, and cuts to higher-dimensional simplicial complexes
\cite{DKM15,mn21}.
Recall that in a graph, flows and circulations are defined on a vector
space over edges, while demands and cuts  are defined on a vector
space over vertices.
On a connected graph, a demand vector is a vector orthogonal to the all-ones
vector, i.e. in the kernel of the boundary operator $\partial_0$.
A flow that routes demand $\dd$, is a vector $\ff$ such that
$\partial_1 \ff = \dd$.

More generally, on a $k$-complex,
we say a demand vector is a vector $\dd$ on $(k-1)$-simplices with $\dd \in
\ker(\partial_{k-1})$.
We say a flow is a vector $\ff$ on $k$-simplices, and that the flows
$\ff$ routes demand $\dd$ if $\partial_k \ff = \dd$.
Given a demand vector $\dd \in
\ker(\partial_{k-1})$ and a capacity vector $\cc$ for the $k$-simplexes,
a reasonable generalization of the max-flow problem to $k$-complexes
is to compute a flow $\ff$ satisfying $\partial_k \ff = \alpha \dd$ 
and ${\bf 0} \le \ff \le \cc$ to maximize the flow value $\alpha$.

\paragraph{Solving Linear Equations.}
Linear equations are ubiquitous in computational mathematics, computer science, engineering, physics, biology, and economics.
Currently, the best known algorithm for solving general dense linear equations in dimensions $n \times n$
runs in time $\tilde{O}(n^{\omega})$, where $\omega < 2.3728596$ is the matrix multiplication constant~\cite{AW21}.
For sparse linear equations with $N$ nonzero coefficients and condition number $\kappa$, 
the best known approximate algorithms run in time $\Otil(\min\{N^{2.27159}, N\kappa\})$, 
where the first runtime is from \cite{PV21,nie21} and the second is by the conjugate gradient\footnote{If the coefficient matrix is symmetric positive semidefinite, 
the runtime is $\tilde{O}(N\sqrt{\kappa})$.} \cite{HS52}.

In contrast to general linear equations, linear equations in graph Laplacians and its generalizations
can be solved asymptotically faster, as mentioned earlier.
In addition, faster solvers are also known for restricted classes of total-variation matrices \cite{KMT11},
stiffness matrices from elliptic finite element systems \cite{BHV08},
and 2 and 3-dimensional truss stiffness matrices~\cite{DS07,KPSZ18}.
An interesting open question is to what extent one can generalize these faster solvers to more 
classes of matrices.

\paragraph{Reduction From Sparse Linear Equations.}
\cite{KWZ20} defines a parameterized family of hypotheses for runtime of solving sparse linear equations.
Under these hypotheses, they prove hardness of approximately solving packing and covering linear programs.
For example, if one can solve a packing linear program up to $\eps$ accuracy 
in time $\tilde{O}(\# \text{ of nonzero coefficients} \times \eps^{-0.165})$, then one can solve 
a system of linear equations in time asymptotically faster than $\tilde{O}(\# \text{ of nonzero coefficients} \times \text{condition number of matrix})$, 
which is the runtime of conjugate gradient.

\subsection{Organization of the Remaining Paper}
In Section \ref{sect:preliminaries}, we present some basic background knowledge related to simplicial homology and systems of linear equations.
In Section \ref{sect:main_results}, we state our main theorem, together with an overview of our proof. 
We first show a reduction for difference-average linear equations to 2-complex boundary operation linear equations 
under the assumption that the right-hand side vector is in the image of the coefficient matrix.
We describe the reduction algorithm and analyze it for exact solvers in Section \ref{sect:construction}, 
and analyze it for approximate solvers in Section \ref{sect:approximate_solvers}.
Then in Section \ref{sect:regression}, we slightly modify the
reduction for approximate solvers to handle the general case, when the
right-hand side vector may not be in the image of the coefficient matrix.
For completeness, we reduce general linear equations to difference-average linear equations in Appendix \ref{sect:appendix_first_reduction} and reduce 2-complex boundary linear equations to combinatorial Laplacians linear equations in Appendix \ref{sect:appendix_lap2boundary}. In Appendix \ref{sect:appendix_IPM}, we present how linear equations in $\partial_{2}\partial_{2}^\top$ are related to Interior Point Methods.


\section{Preliminaries}
\label{sect:preliminaries}

\subsection{Simplicial Homology}
\label{sect:prelim_simplicial}

We define the basic concepts of simplicial homology. We recommend the readers the books~\cite{munkres18}
and~\cite{hatcher00} for a more complete treatment.

\paragraph{Simplicial Complexes.}
A \emph{$k$-dimensional simplex} (or $k$-simplex) $\sigma = \conv \{v_0, \ldots, v_k\}$ is the convex hull of $k+1$ affinely independent points $v_0, \ldots, v_k$.
For example, 0,1,2-simplexes are vertices, edges, and triangles, respectively.
A \emph{face} of $\sigma$ is the convex hull of a non-empty subset of $\{v_0, v_1,\ldots, v_k\}$. %
An \emph{orientation} of $\sigma$ is given by an ordering $\Pi$ of its vertices, 
written as $\sigma = [v_{\Pi(0)}, \ldots, v_{\Pi(k)}]$, such that two orderings define 
the same orientation if and only if they differ by an even permutation.
If $\Pi$ is even, then $[v_{\Pi(0)}, \ldots, v_{\Pi(k)}]=[v_0,\ldots, v_k]$;
if $\Pi$ is odd, then $[v_{\Pi(0)}, \ldots, v_{\Pi(k)}]=-[v_0,\ldots, v_k]$.
 
A \emph{simplicial complex} $\mathcal{K}$ is a finite collection of simplexes such that (1) for every $\sigma\in\mathcal{K}$ if $\tau\subset \sigma$ then $\tau\in \mathcal{K}$ 
and (2) for every $\sigma_1, \sigma_2\in\mathcal{K}$, $\sigma_1\cap\sigma_2$ is either empty or a face of both $\sigma_1, \sigma_2$. %
The \emph{dimensions} of $\calK$ is the maximum dimension of any simplex in $\calK$.
We refer to a simplicial complex in $k$ dimensions as a $k$-complex.

\paragraph{Boundary Operators.}

A $k$-chain is a formal sum of the oriented $k$-simplices in $\calK$ 
with the coefficients over $\mathbb{R}$.
Let $C_k(\mathcal{K})$ denote the $k$th chain space.
The \textit{boundary operator} is a linear map $\partial_k: C_k(\mathcal{K})\rightarrow C_{k-1}(\mathcal{K})$ 
such that for an oriented $k$-simplex $\sigma = [v_0, v_1, \ldots, v_k]$,
\[
	\partial_k(\sigma)=\sum_{i=0}^{k}(-1)^i[v_0,\ldots,\hat{v}_i,\ldots, v_k],
\]
where $[v_0,\ldots,\hat{v}_i,\ldots, v_k]$ is the oriented $(k-1)$-simplex obtained by removing $v_i$ from $\sigma$,
and $(-1)^i$ is its \textit{induced orientation}.
The operator $\partial_k$ can be written as a matrix in $n_{k-1} \times n_k$ dimensions, 
where $n_{d}$ is the number of $d$-simplices in $\calK$.
The $(i,j)$th entry of $\partial_k$ is $\pm 1$ if the $i$th $(k-1)$-simplex is a face of the $j$th $k$-simplex 
where the sign is determined by the orientations, and $0$ otherwise.
See Figure \ref{fig:partial_example} and Eq. \eqref{eq:boundary_example} for an example.
\begin{figure}[!h] 
	\centering 
	\includegraphics[width=0.3\textwidth]{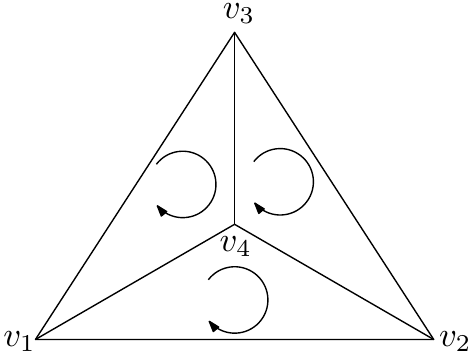} 
	\caption{An example of boundary operator and oriented triangulation. 
		We set a clockwise orientation for 2-simplices, and set the orientation for 1-simplices as the order of increasing vertex indices.}
	\label{fig:partial_example} 
\end{figure}
\begin{equation}
	\label{eq:boundary_example}
	\footnotesize
	\partial_2=\begin{bNiceArray}{ccc}[first-row,first-col]
		& [v_1,v_4,v_2] & [v_2,v_4,v_3] & [v_1,v_3,v_4] \\
		[v_1,v_2] & -1 & 0 & 0 \\
		[v_2,v_3] & 0 & -1 & 0 \\
		[v_1,v_3] & 0 & 0 & 1 \\
		[v_1,v_4] & 1 & 0 & -1 \\
		[v_2,v_4] & -1 & 1 & 0 \\
		[v_3,v_4] & 0 & -1 & 1 \\
	\end{bNiceArray}.
\end{equation}

An important property of the boundary operator is that applying the boundary operator twice results in the zero operator, i.e.,
\begin{equation}
	\label{eq:two_boundary}
	\partial_{k-1}\partial_{k}=\mathbf{0}.
\end{equation}
This implies $\im(\partial_{k})\subseteq \ker(\partial_{k-1})$.
Thus,  we can define the quotient space $H_k = \ker(\partial_{k}) \setminus \im(\partial_{k+1})$, 
referred to as the $k$th homology space of $\calK$.
The dimension of $H_k$ is the $k$th Betti number of $\calK$, which plays an important role 
in understanding the homology spaces.

\paragraph{Hodge Theory and Combinatorial Laplacians.} 

Combinatorial Laplacians arise from the discrete Hodge decomposition.
\begin{theorem}[Hodge decomposition \cite{lim20}]
	Let $\AA\in\R^{m\times n}$ and $\BB\in\R^{n\times p}$ be matrices satisfying $\AA\BB=\mathbf{0}$. 
	Then, there is an orthogonal direct sum decomposition
	\[\R^n=\im(\AA^\top)\oplus\ker (\AA^\top\AA+\BB\BB^\top)\oplus\im(\BB).\]
\end{theorem}	
By Eq. \eqref{eq:two_boundary}, it is valid to set $\AA=\partial_{k}$ and $\BB=\partial_{k+1}$. 
The matrix we get in the middle term is the \emph{combinatorial Laplacian}:
$$\calL_k\defeq\partial_{k}^\top \partial_{k}+\partial_{k+1} \partial_{k+1}^\top$$ 
In particular, $\calL_0 = \partial_1 \partial_1^\top$ is the graph Laplacian.
The $k$th homology space $H_k(\calK)$ is isomorphic to $\ker(\calL_k)$, 
and thus the $k$th Betti number of $\calK$ equals the dimension of $\ker(\calL_k)$.

\paragraph{Triangulation.}

A \textit{triangulation} of a topological space $\mathcal{X}$ is a
simplicial complex $\mathcal{K}$ together with a homeomorphism between
$\mathcal{X}$ and $\mathcal{K}$.
In this paper, the only topological spaces that we compute
triangulations of are 2-dimensional manifolds.
A 2-dimensional manifold can be triangulated by a 2-complex, where every edge in the 2-complex is contained in exactly one triangle (boundary edge) or two triangles (interior edge).
An \textit{oriented triangulation} of a 2-dimensional manifold is a
triangulation together with an orientation for each triangle such that
any two neighboring triangles induce opposite signs on their shared
interior edge.

Figure \ref{fig:partial_example} is an example of (oriented) triangulation: the topological space is a disk; boundary edges are $[v_1,v_2], [v_2,v_3], [v_1,v_3]$; interior edges are $[v_1,v_4], [v_2,v_4], [v_3,v_4]$; the orientation for each triangle is clockwise.

\subsection{Notation for Matrices and Vectors}

We use parentheses to denote entries of a matrix or a vector: 
Let $\AA(i,j)$ bet the $(i,j)$th entry of a matrix $\AA$, and let $\xx(i)$ bet the $i$th entry of a vector $\xx$. We use $\vecone_n, \mathbf{0}_n$ to denote $n$-dimensional all-one vector and all-zero vector, respectively. 
We define $\norm{\xx}_{\max}=\max_{i\in [n]} \abs{\xx(i)}$, $\norm{\xx}_1=\sum_{i\in[n]}\abs{\xx(i)}$.
Given a matrix $\AA\in\R^{d\times n}$, we use $\AA(i)$ to denote the $i$th row of $\AA$ and $\nnz(\AA)$ the number of nonzero entries of $\AA$. 
Without loss of generality, we assume that $\nnz(\AA)\geq \max\{d,n\}$. We let $\norm{\AA}_{\max}=\max_{i,j}\abs{\AA(i,j)}$. 
We use $\im(\AA)$ to denote the image (i.e., the column space) of $\AA$ and $\Null(\AA)$ the null space of $\AA$.
We let $\PPi_{\AA}= \AA(\AA\AA^\top)^\dagger \AA^\top$ be the orthogonal projection onto $\im(\AA)$, where $\MM^\dagger$ is the pseudo-inverse of $\MM$.
Let $\lambda_{\max}(\AA)$ be the maximum eigenvalue of $\AA$ and $\lambda_{\min}(\AA)$
the minimum \emph{nonzero} eigenvalue of $\AA$.
Similarly, let $\sigma_{\max}(\AA)$ be the maximum eigenvalue of $\AA$ and $\sigma_{\min}(\AA)$
the minimum \emph{nonzero} singular value of $\AA$.
The condition number of $\AA$, denoted by $\kappa(\AA)$, 
is the ratio of the maximum to the minimum \emph{nonzero} singular value of $\AA$.

We define a function $U$ that takes a matrix $\AA$ and a vector $\bb$ as arguments
and returns the maximum of $\norm{\cdot}_{\max}$ of all the arguments, that is, 
\[
U(\AA,\bb) = \max\{\norm{\AA}_{\max}, \norm{\bb}_{\max}\}.	
\]

\subsection{Systems of Linear Equations}
\label{sect:prelim_linear_system}

We define approximately solving linear equations in a general form, following \cite{KZ17}.
For more details, we refer the readers to Section 2.1 of \cite{KZ17}.

\begin{definition}[Linear Equation Problem (\ls)]
	\label{def:ls}
	Given a matrix $\AA \in \R^{d \times n}$, a vector $\bb \in \R^{d}$, we refer to the linear equation problem for the tuple $(\AA,\bb)$, 
	denoted by \ls~$(\AA, \bb)$, as the problem of finding an $\xx\in\R^n$ such that 
	\[\xx \in \argmin_{\xx\in\R^n}\norm{\AA\xx-\bb}_2.\]
\end{definition}

\begin{fact}
Let $\xx^* \in \argmin_{\xx\in\R^n}\norm{\AA\xx-\bb}_2$. Then, 
$$\AA\xx^* = \AA (\AA^\top \AA)^{\dagger} \AA^\top \bb = \PPi_{\AA} \bb$$
and 
\[
\norm{\AA\xx^* - \bb}_2^2 = \norm{(\II - \PPi_{\AA})\bb}_2^2	
\]
\end{fact}

By the above fact,
solving \ls~$(\AA,\bb)$ is equivalent to finding an $\xx$ such that $\AA\xx = \PPi_{\AA}\bb$.
This equation is known as the normal equation, and it is always feasible.
If $\bb \in \im(\AA)$, then $\PPi_{\AA} \bb = \bb$.

In practice, we are more interested in \emph{approximately} solving linear equations, 
since numerical errors are unavoidably in data collection and computation 
and approximate solvers may run faster.

\begin{definition}[Linear Equation Approximation Problem (\lsa)]
	\label{def:lsa}
	Given a matrix $\AA \in \R^{d \times n}$, vectors $\bb \in
        \R^{d}$, and an error parameter $\epsilon\in (0,1]$, we refer to linear equation approximate problem for the tuple $(\AA,\bb,\epsilon)$,
		denoted by \lsa~$(\AA,\bb,\eps)$,
		 as the problem of finding an $\xx\in\R^n$ such that
	\[\norm{\AA\xx-\PPi_{\AA}\bb}_2\leq \eps\norm{\PPi_{\AA}\bb}_2.\]
\end{definition}

\begin{fact}
	\label{lm:exact_approximate}
	Let $\xx$ be a solution to \lsa~$(\AA,\bb,\epsilon)$. Then,
	\[
		\norm{\AA\xx-\bb}_2^2\leq \norm{\AA\xx^*- \bb}_2^2+\epsilon^2\norm{\PPi_{\AA}\bb}_2^2.
		\]
\end{fact}

The definition of the approximate error in Definition~\ref{def:lsa} is equivalent to several error notions 
that are commonly used in solving linear equations. In particular,
\begin{align*}
	\norm{\AA \xx - \PPi_{\AA} \bb}_2 =
	\norm{\AA^\top \AA \xx - \AA^\top \bb}_{(\AA^\top \AA)^{\dagger}}
	= \norm{\xx - \xx^*}_{\AA^\top \AA}.
\end{align*}

\subsubsection{Matrix Classes}

We are interested in linear equations whose coefficient matrices belonging to
the following matrix classes.
\begin{enumerate}
	\item $\mathcal{G}$ refers to the class of \textit{General Matrices} that have integer entries
	and do not have all-0 rows and all-0 columns.
	We refer to linear equations whose coefficient matrix is in $\mathcal{G}$ 
	as \textit{general linear equations}.
	\item $\mathcal{DA}$ refers to the class of \textit{Difference-Average Matrices} whose rows fall into two categories:
	\begin{enumerate}
		\item A \textit{difference row} which has exactly two nonzero entries $1$ and $-1$;
		\item An \textit{average row} which has exactly three nonzero entries  $1$, $1$, and $-2$. 
	\end{enumerate}
	Multiplying a difference row vector to a column vector $\xx$ gives $\xx(i) - \xx(j)$;
	multiplying an average row vector to $\xx$ gives $\xx(i) + \xx(j) - 2\xx(k)$.
	We refer to linear equations whose coefficient matrix is in $\mathcal{DA}$ 
	as \textit{difference-average linear equations}.
	\item $\mathcal{B}_{2}$ refers to the class of \textit{Boundary Operator Matrices $\partial_2$ in 2-complexes}.
	We refer to linear equations whose coefficient matrix is in $\mathcal{B}_2$ 
	as \textit{2-complex boundary linear equations}.
\end{enumerate}

Our definition of ``general matrices'' specifies the matrix must have integer entries.
However, when the input matrix is invertible, using a simple rounding
argument, we can convert any linear equation into an linear equation with
integer entries $\Otil(1)$ bits per
entry.
  We caution the reader this relies on our
  definition of $\Otil(\cdot)$ as hiding polylogarithmic factors in
  the input condition number.
  In general, the condition number can be exponentially large --
  however, our results are mainly of interest when the condition number
  is quasipolynomially bounded.

\subsubsection{Reduction Between Linear Equations}

We will again follow the definition of efficient reductions in~\cite{KZ17}.
We say \lsa~over matrix class $\calM_1$ is \emph{nearly-linear time reducible} to 
\lsa~over matrix class $\calM_2$, denoted by $\calM_1 \le_{nlt} \calM_2$, if the following holds:
\begin{enumerate}
	\item There is an algorithm that maps 
	an arbitrary instance \lsa~$(\MM_1, \cc_1, \epsilon_1)$ where $\MM_1 \in \calM_1$ 
	to an instance \lsa~$(\MM_2, \cc_2, \epsilon_2)$ where $\MM_2 \in \calM_2$ such that 
	there is another algorithm that can map a solution to \lsa~$(\MM_2,  \cc_2, \epsilon_2)$ to a solution 
	to  \lsa~$(\MM_1,  \cc_1, \epsilon_1)$.
	\item Both the two algorithms run in time $\Otil(\nnz(\MM_1))$.
	\item In addition, we can guarantee
	$\nnz(\MM_2) = \Otil{(\nnz(\MM_1))}, \text{ and } \eps_2^{-1},\kappa(\MM_2),U(\MM_2,\bb_2)
	= \poly(\nnz(\MM_1), \eps_1^{-1},  \kappa(\MM_1), U(\MM_1,\bb_1)).$
\end{enumerate}

We do not require a nearly-linear time reduction to preserve the
number of variables or constraints (dimensions) of a system of linear equations.
The dimensions of the new linear equation instance that we construct can be much larger than that 
of the original instance. 
On the other hand, a reduction that \emph{only} preserves dimensions may construct a dense linear equation instance even if the original instance is sparse.
A nearly-linear time reduction that preserves \emph{both} the number of nonzeros and dimensions 
would be stronger than what we achieve.

\begin{fact}
	If $\calM_1 \le_{nlt} \calM_2$ and $\calM_2 \le_{nlt} \calM_3$, 
	then $\calM_1 \le_{nlt} \calM_3$.
\end{fact}

\begin{definition}[Sparse linear equation complete (\slec)]
	We say \lsa~over a matrix class $\calM$ is \emph{sparse-linear-equation-complete}
	if $\calG \le_{nlt} \calM$.
\end{definition}

\begin{fact}
	Suppose \lsa~over $\calM$ is \slec.
	If one can solve all instances \lsa~$(\AA,\bb,\eps)$ with $\AA \in \calM$ in time $\tilde{O}(\nnz(\AA)^c)$ where $c \ge 1$,
	then one can solve all instances \lsa~$(\AA',\bb',\eps')$ with $\AA' \in \calG$ in time $\tilde{O}(\nnz(\AA')^c)$.
	\label{fact:slec}
\end{fact}

Under the above definitions,
\cite{KZ17} implicitly shows the following results. 
We provide an explicit and simplified proof in Appendix~\ref{sect:appendix_first_reduction}.

\begin{theorem}[Implicitly stated in~\cite{KZ17}]
	\lsa~over $\mathcal{DA}$ is \slec.
	\label{thm:prelim_kz17}
\end{theorem}

\def\epsda{\epsilon^{DA}}
\def\epsb{\epsilon^{B_2}}

\section{Main Results}
\label{sect:main_results}

Our main result is stated in the following theorem.

\begin{theorem}
  \label{thm:boundarySLEC}
\lsa~over $\mathcal{B}_2$ is \slec.
\end{theorem}

Although our main theorem focuses on linear equation approximate problems, we construct nearly-linear time reductions for both linear equation problem \ls~and its approximate counterpart \lsa. 
We first reduce \ls~instances $(\AA,\bb)$ (and \lsa~instances $(\AA,\bb,\epsilon)$)
over difference-average matrices 
to those over 2-complex boundary operator matrices, 
\emph{under the assumption $\bb \in \im(\AA)$} (stated in Theorem \ref{thm:main_feasible} and \ref{thm:main_feasible_approx}). 
In this case, the constructed 2-complexes have unit edge weights. 
We then provide a slightly modified nearly-linear time reduction 
for \lsa~$(\AA,\bb,\eps)$ over difference-average matrices to \lsa~over 2-complex boundary operator matrices
\emph{without assuming $\bb \in \im(\AA)$} (stated in Theorem \ref{thm:main_general}). In this case, we introduce polynomially bounded edge weights for the constructed 2-complexes.

\begin{theorem} 
	\label{thm:main_feasible}
	Given a linear equation instance \ls~$(\AA,\bb)$ where $\AA\in \mathcal{DA}$ and $\bb\in\im(\AA)$, 
	we can reduce it to an instance \ls~$(\partial_{2},\gamma)$ where $\partial_2 \in \mathcal{B}_2$, in time $O(\nnz(\AA))$, 
	such that a solution to \ls~$(\partial_2, \gamma)$ can be mapped to a solution to \ls~$(\AA,\bb)$ in time $O(\nnz(\AA))$.
	
\end{theorem}

\begin{theorem} 
	\label{thm:main_feasible_approx}
	Given a linear equation instance \lsa~$(\AA,\bb, \epsda)$ where $\AA\in \mathcal{DA}$ and $\bb\in\im(\AA)$, 
	we can reduce it to an instance \lsa~$(\partial_{2},\gamma, \epsb)$ where $\partial_2 \in \mathcal{B}_2$ 
	and $\epsb \le \frac{\epsda}{42 \nnz(\AA)}$, in time $O(\nnz(\AA))$, 
	such that a solution to \lsa~$(\partial_2, \gamma, \epsb)$ can be mapped to a solution to \lsa~$(\AA,\bb, \epsda)$ in time $O(\nnz(\AA))$.
	
\end{theorem}

\begin{theorem}
	\label{thm:main_general}
	Given an instance \lsa~$(\AA,\bb,\epsda)$ where $\AA\in \mathcal{DA}$, 
	we can reduce it to an instance \lsa~$(\WW^{1/2}\partial_{2},\WW^{1/2}\gamma, \epsb)$ 
	where $\partial_2 \in \mathcal{B}_2$ and $\WW$ is a diagonal matrix with nonnegative diagonals, in time $O(\nnz(\AA))$.
	Let $s, \eps, K,U$ denote $\nnz(\AA), \epsda, \kappa(\AA), U(\AA,\bb)$, respectively. 
	Then, we can guarantee that
	\begin{align*}
		& \nnz(\partial_{2})=O(s), ~ U(\WW^{1/2} \partial_2,\WW^{1/2} \gamma) = O\left(s U \eps^{-1}  \right), \\
		& ~ \epsb = \Omega(\eps U^{-1} s^{-1} ),
		~ \kappa(\WW^{1/2}\partial_2) = O \left( s^{15/2} K^2 \eps^{-2} \right)
	\end{align*}
	and a solution to \lsa~$(\WW^{1/2}\partial_{2}, \WW^{1/2}\gamma, \epsb)$ can be mapped to a solution to 
	\lsa~$(\AA,\bb,\epsda)$ in time $O(\nnz(\AA))$.
\end{theorem}

We will prove Theorem \ref{thm:main_feasible} in Section \ref{sect:construction}, 
Theorem \ref{thm:main_feasible_approx} in Section \ref{sect:approximate_solvers}, 
and Theorem \ref{thm:main_general} in Section \ref{sect:regression}.

\subsection{Overview of Our Proof}

Multiplying a 2-complex boundary operator $\partial_2 \in \mathbb{R}^{m \times t}$ to a vector $\ff \in \mathbb{R}^t$ 
can be interpreted as transforming flows in the triangle space to demands in the edge space.
Given $\dd \in \mathbb{R}^d$, solving $\partial_2 \ff = \dd$ can be interpreted as finding flows in the triangle space 
subject to edge demands in $\dd$.
We will encode difference-average linear equations as a 2-complex flow network.

\paragraph{Encoding a Single Equation.}
We observe a simple fact:
If we glue two triangles $\Delta_1, \Delta_2$ with the same orientation,
then the net flow $\partial_2 \ff$ on the shared interior edge is $\ff(\Delta_1) - \ff(\Delta_2)$ (see Figure \ref{fig:proof_overview}~(a)); 
if we glue two triangles $\Delta_1, \Delta_2$ with opposite orientations,
then the net flow $\partial_2 \ff$ on the shared interior edge is $\ff(\Delta_1) + \ff(\Delta_2)$ (see Figure \ref{fig:proof_overview}~(b)).
Given an equation $\aa^\top \xx = b$ with the nonzero coefficients being $\pm 1$, 
we can encode it by gluing more triangles as above and setting the demand of the shared interior edge to be $b$.
To handle the coefficient $-2$ in an average equation, say $\xx(i) + \xx(j) - 2\xx(k)$, 
we implicitly interpret it as $\xx(i) + \xx(j) -  \xx(k_1) - \xx(k_2)$ 
together with an additional difference equation $\xx(k_1) = \xx(k_2)$ (see Figure \ref{fig:proof_overview}~(c)).

\begin{figure}[!h] 
	\centering 
	\includegraphics[width=\textwidth]{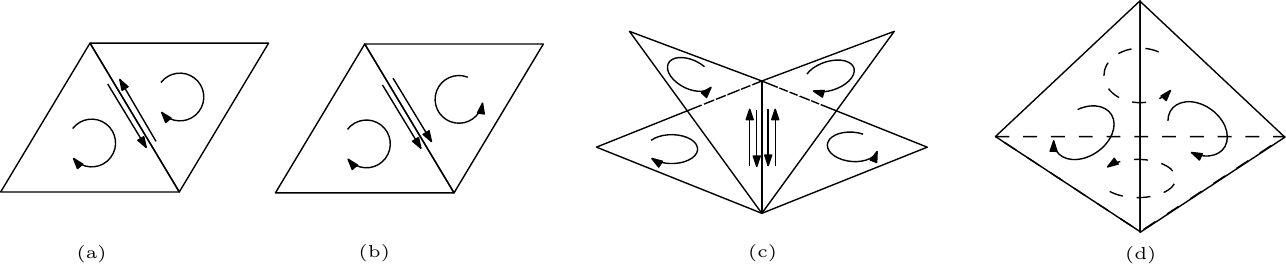} 
	\caption{An illustration for encoding a single equation and encoding a variable.}
	\label{fig:proof_overview} 
\end{figure}

\paragraph{Encoding a Variable.}
We use a sphere to encode a variable involved in many equations. 
We can obtain an oriented triangulation of the sphere and set all the edge demand to be 0 
so that all the triangles on the sphere must have an equal flow value (see Figure \ref{fig:proof_overview}~(d)).

\paragraph{Putting All Together.}
For each variable $\xx(i)$ and the sphere for $\xx(i)$, we create a ``hole'' for each equation that involves $\xx(i)$, 
and then attach a tube. We can have an oriented triangulation of the tubes so that the triangles on the tubes have equal value as the triangles on the sphere.
We then connect these tubes properly to encode each given difference and average equation.

\paragraph{Discussion.}
\begin{itemize}

	\item \textit{Why encode difference/average equations rather
            than directly encoding general equations with integer coefficients?}\\
	We can generalize the above encoding method to encode a
        general equation $\gg^\top\yy=c$ with arbitrary integer
        coefficients into a 2-complex with roughly $\norm{\gg}_1$
        tubes.
        However, the encoding size required to express a general system of linear
        equations $\GG\yy=\cc$ this way can be as large as
        $\Omega(\nnz(\GG)\norm{\GG}_{\max})$.
        This dependence on $\norm{\GG}_{\max}$ is
        prohibitive, and makes for a fairly weak result.
	
	On the other hand, we can first reduce the general linear
        equations $\GG\yy=\cc$ into difference-average linear
        equations $\AA\xx=\bb$, where $\norm{\AA}_{\max}=2$ and
        $\nnz(\AA)=O\left(\nnz(\GG)\log \norm{\GG}_{\max}\right)$ (by
        Lemma \ref{lem:GZToMC2NNZ}). Then we can encode $\AA\xx=\bb$ into
        a 2-complex. The encoding size required to express the the difference-average
        linear equations as a 2-complex is thus $O(\nnz(\AA))$ (by Lemma
        \ref{lem:size_boundary}). Thus, the overall encoding size
        required to express the original linear equation $\GG\yy=\cc$
        is now $\Otil\left(\nnz(\GG)\right)$, exponentially improving
        the dependence on $\norm{\GG}_{\max}$.
      
	Therefore, the two-step reduction is a nearly-linear time
        reduction while the one-step reduction is not.

      \item \textit{Why encode into a 2-complex rather than a
          1-complex?}\\
        We do not expect that general linear equations with integer
        coefficients can be efficiently encoded using a
        1-complex. This would immediately imply a nearly-linear time
        solver for general linear equations, as fast solvers for 1-complex
        operators exist (using Laplacian linear equation solvers).
        
\end{itemize}

%
\section{Reducing Exact Solvers for $\da$ to $\b2$ Assuming the Right-Hand Side Vector in the Image of the Coefficient Matrix}
\label{sect:construction}

In this section, we describe a nearly-linear time reduction from instances \ls~$(\AA,\bb)$ over $\da$
to instances \ls~over $\b2$, under the assumption that $\bb \in \im(\AA)$.
In Section \ref{sect:approximate_solvers}, we will show that the same reduction with a 
carefully chosen error parameter reduces 
linear equation approximate problem \lsa~over $\da$ to \lsa~over $\b2$, assuming $\bb \in \im(\AA)$.
 In Section~\ref{sect:regression}, we will slightly modify the reduction to drop the assumption $\bb \in \im(\AA)$
for \lsa.

Recall that an instance \ls~$(\AA,\bb)$ over $\da$
only consists of two types of linear equations:
\begin{enumerate}
	\item \textit{Difference equation}: $\xx(i)-\xx(j)=\bb(q)$,
	\item \textit{Average equation}: $\xx(i)+\xx(j)-2\xx(k)=0$.
\end{enumerate}

Suppose \ls~$(\AA,\bb)$ has $d_1$ difference equations and $d_2$ average equations.
Without loss of generality, 
we reorder all the equations so that the first $d_1$ equations are difference equations
and the rest are average equations.

\subsection{Reduction Algorithm}

Given an instance \ls~$(\AA,\bb)$ where $\AA$ is a $d \times n$ matrix in $\mathcal{DA}$, 
the following algorithm \textsc{Reduce$\da$To$\b2$} constructs a 2-complex and a system of linear equations in its boundary operator.\\\\
\noindent\textbf{Algorithm \textsc{Reduce$\da$To$\b2$}}\\
\textbf{Input}: an instance \ls~$(\AA,\bb)$ where $\AA\in\da$ is a $d\times n$ matrix and $\bb \in \mathbb{R}^d$.\\
\textbf{Output}: $(\partial_2, \gamma, \boldsymbol{\Delta}^c)$ where $\partial_2\in \b2$ is an $m \times t$ matrix, $\gamma\in\R^m$, and $\triangle^c$ is a set of $n$ triangles.
\begin{enumerate}
	\item For each $i \in [n]$ and variable $\xx(i)$ in \ls~$(\AA,\bb)$, we construct a sphere $\calS_i$.
	\item For each $q \in [d_1]$, which corresponds to a difference equation $\xx(i) - \xx(j) = \bb(q)$, we add a \textit{loop} $\alpha_q$ with a net flow demand $\bb(q)$. Then,
	\begin{enumerate}
		\item we add a boundary component $\beta_{q,i}$ on $\calS_i$, and a boundary component $\beta_{q,j}$ on $\calS_j$;
		\item we construct a tube $\calT_{q,i}$ with boundary components $\{-\beta_{q,i}, \alpha_q \}$, and a tube $\calT_{q,j}$ with boundary components $\{-\beta_{q,j}, -\alpha_q \}$. 
	\end{enumerate}
	See Figure \ref{fig:construction1} for an illustration.
	\footnote{Note that since the loop $\alpha_q$ has demand $\bb(q)$, our construction is different from identifying the boundary component $\alpha_q$ of $\calT_{q,i}$
		and the boundary component $-\alpha_q$ of $\calT_{q,j}$.}
	\begin{figure}[!h] 
		\centering 
		\includegraphics[width=0.6\textwidth]{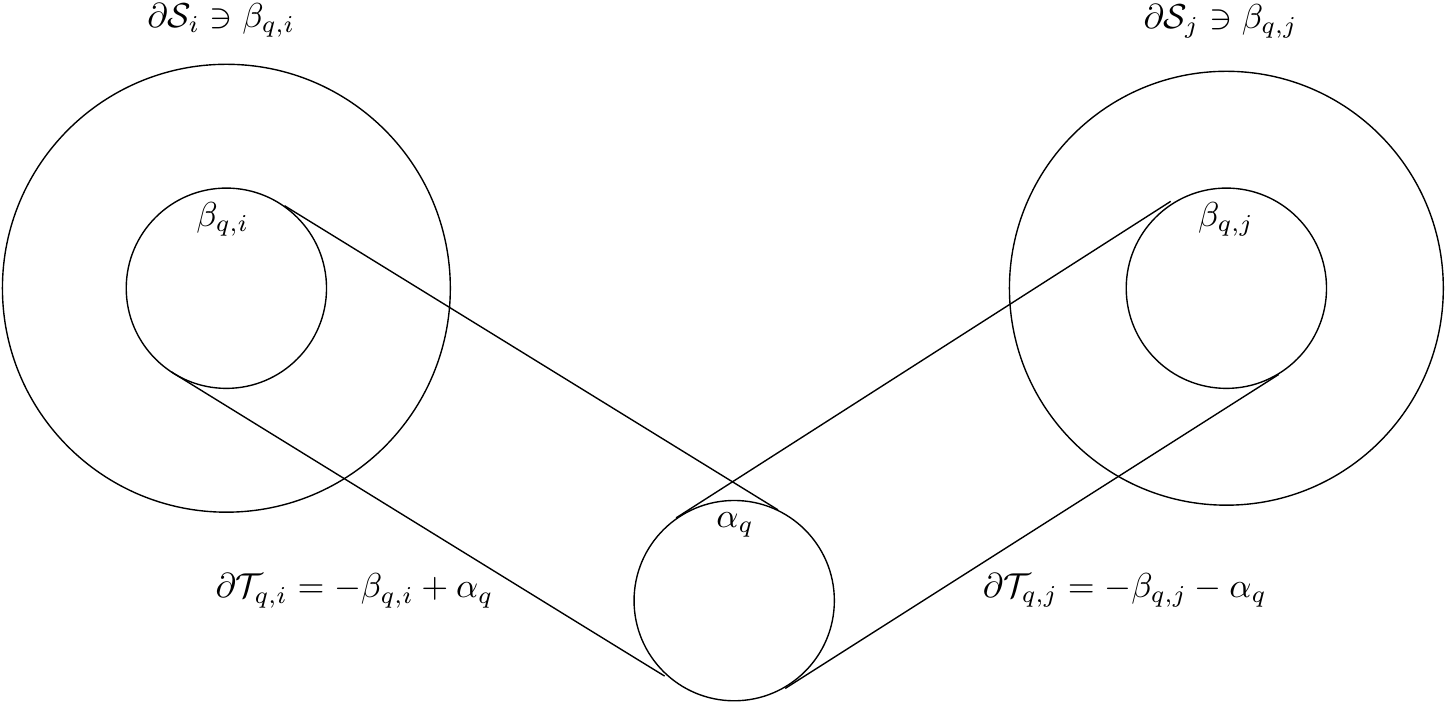} 
		\caption{The construction for a difference equation $\xx(i)-\xx(j)=\bb(q)$.}
		\label{fig:construction1} 
	\end{figure}
	\item For each $q \in  \{d_1+1,\ldots, d \}$, which corresponds to an average equation $\xx(i) + \xx(j) - 2\xx(k) = \bb(q) = 0$, 
	we add a loop $\alpha_q$ with zero net flow demand. Then,
	\begin{enumerate}
		\item we add a boundary component $\beta_{q,i}$ on $\calS_i$, a boundary component $\beta_{q,j}$ on $\calS_j$,
		and two boundary components $\beta_{q,k,1}, \beta_{q,k,2}$ on $\calS_k$; 
		\item we construct a tube $\calT_{q,i}$ with boundary components $\{-\beta_{q,i}, \alpha_q \}$,
		a tube $\calT_{q,j}$ with boundary components $\{-\beta_{q,j}, \alpha_q \}$,
		and two tubes $\calT_{q,k,1}, \calT_{q,k,2}$ with boundary components $\{-\beta_{q,k,1}, -\alpha_q \}$ and $\{-\beta_{q,k,2}, -\alpha_q \}$, respectively.
	\end{enumerate}
	See Figure \ref{fig:construction2} for an illustration\footnote{
		As four tubes are connected to a single loop, to avoid the intersection of tubes before attaching the loop, a higher-dimensional space is required. }.
	\begin{figure}[!h] 
		\centering 
		\includegraphics[width=0.6\textwidth]{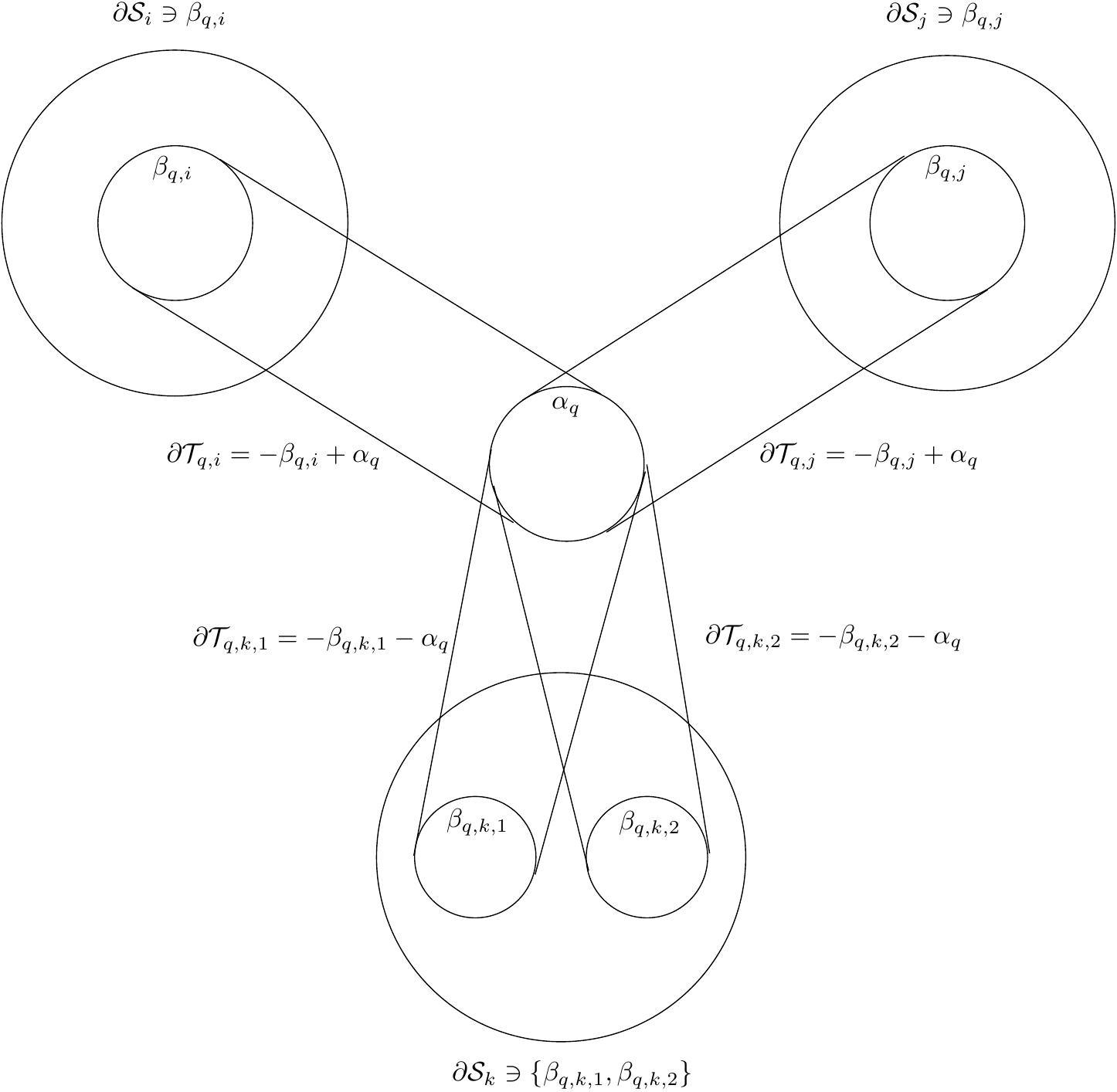} 
		\caption{The construction for an average equation $\xx(i)+\xx(j)-2\xx(k)=0$.}
		\label{fig:construction2} 
	\end{figure}
	\item For each $i \in [n]$, the punctured sphere $\calS_i$ and the tubes connected to $\calS_i$ form a continuous topological space.
	We construct an oriented triangulation for this space such that
	the induced orientation of each edge on a loop $\alpha_q$ is consistent with the orientation of $\alpha_q$.
	We will describe this oriented triangulation subroutine in Section \ref{sec:oriented_triangulation}.
	Let $\calK$ be the oriented 2-complex.
	Let $\partial_2$ be the boundary operator of $\calK$.

	\item Each edge on a loop $\alpha_q$
	has net demand $\bb(q)$;
	each other edge
	has net demand $0$.
	Let $\gamma$ be the vector of the net flow demands.	

	\item On each triangulated sphere $\calS_i$, we choose an arbitrary triangle $\Delta_i \in \calS_i$ as the \emph{central triangle}. 
	Let $\triangle^c$ be the set of all the central triangles.

	\item We return $(\partial_2, \gamma, \triangle^c)$.
\end{enumerate}

The following algorithm \textsc{MapSoln$\b2$to$\da$} maps a solution $\ff$ to \ls~$(\partial_2, \gamma)$ to a solution $\xx$ to \ls~$(\AA,\bb)$.\\\\
\noindent\textbf{Algorithm \textsc{MapSoln$\b2$to$\da$}}\\
\textbf{Input}: a tuple $(\AA, \bb, \ff, \boldsymbol{\Delta}^c)$, where $\AA \in \da$ is a $d\times n$ matrix, $\bb \in \mathbb{R}^d$, $\ff \in \mathbb{R}^t$, and $\boldsymbol{\Delta}^c$ is the set of $n$ central triangles.\\
\textbf{Output}: a vector $\xx \in \mathbb{R}^n$.
\begin{enumerate}
	\item If $\AA^\top \bb = 0$, we return $\xx = \bf 0$.
	\item Otherwise, we set $\xx(i) = \ff(\Delta_i)$, where $\Delta_i\in \triangle^c$ is the central triangle on sphere $\calS_i$.
\end{enumerate}

\subsubsection{Oriented Triangulation for Punctured Spheres and Tubes}
\label{sec:oriented_triangulation}

We provide a concrete triangulation subroutine here for the benefit of analyzing our reduction algorithm.

\paragraph*{Oriented Triangulation for Punctured Spheres.}	
By our construction, each sphere $\calS_i$ has $b_i = \sum_{q=1}^{d} \abs{\AA(q,i)}$ boundary components.
We will create $\ttil_i$ triangles and $\mtil_i$ edges on $\calS_i$, based on $b_i$.
\begin{enumerate}
	\item If $b_i=1$ (see Figure \ref{fig:triangulation_sphere} (a)), the punctured sphere is topologically equivalent to a disk.
	In this case, $\calS_i$ can be triangulated using a single triangle $[v_{(1)}^1,v_{(1)}^2,v_{(1)}^3]$, thus $\ttil_i=1, \mtil_i=3$. 

	\item If $b_i=2$ (see Figure \ref{fig:triangulation_sphere} (b)), the punctured sphere is topologically equivalent to an annulus. We subdivide the triangle $[v_{(1)}^1,v_{(1)}^2,v_{(1)}^3]$ obtained in the previous case by adding 6 interior edges between vertices of the inner and the outer boundaries: $[v_{(1)}^1, v_{(2)}^1],[v_{(1)}^1,v_{(2)}^2], [v_{(1)}^2, v_{(2)}^1],[v_{(1)}^2,v_{(2)}^3], [v_{(1)}^3, v_{(2)}^2], [v_{(1)}^3, v_{(2)}^3]$, thus $\ttil_i=6, \mtil_i=12$.

	\item Generally, if $b_i=k$ (see Figure \ref{fig:triangulation_sphere} (c)), we subdivide the rightmost triangle $[v_{(1)}^1,v_{(1)}^2,v_{(k-1)}^1]$ obtained in the case of $b_i=k-1$ with the same method. By induction, we have
	\begin{equation}
		\label{eq:sphere_triangulation}
		\ttil_i=5b_i-4, \qquad \mtil_i=9b_i-6, \qquad \text{for}~ b_i\geq 1.
	\end{equation}
\end{enumerate}

\begin{figure}[!h] 
	\centering 
	\includegraphics[width=0.9\textwidth]{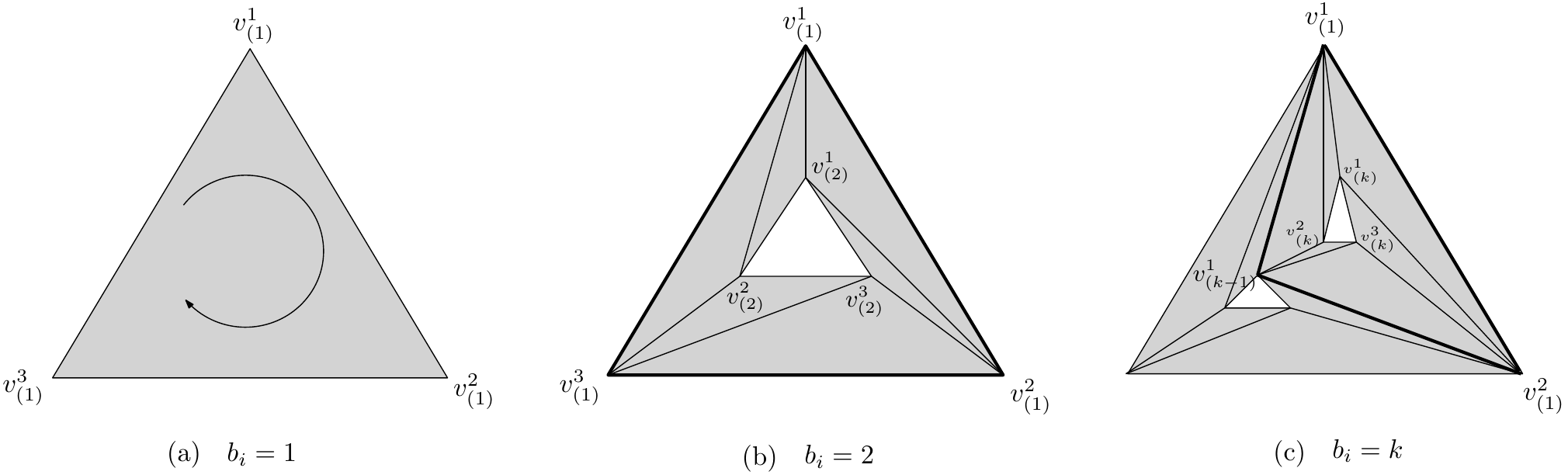} 
	\caption{Oriented triangulation of punctured spheres. The light area represents the ``holes'' defined by boundary components.}
	\label{fig:triangulation_sphere} 
\end{figure}

The orientation for triangles on the same sphere should be identical.
Without loss of generality, we orient all triangles clockwise. Note that with this triangulation method, all boundary components are composed of 3 edges.

\paragraph*{Oriented Triangulation for Tubes.}	
A tube is defined by two boundary components. By our construction, for every tube connected to $\calS_i$, one of the two boundary components is always $-\beta_{q,i,*}$ \footnote{We introduce a third element $*\in\{1,2\}$ in the subscript of $\beta_{q,k,*}$, which is activated only when $\AA(q,k)=-2$.}, and the other one is $\pm\alpha_q$, whose orientation depends on the sign of the entry $\AA(q,i)$. Without loss of generality, we orient anti-clockwise for all $\alpha_q$, thus clockwise for all $-\alpha_q$. 
Therefore, there are two possibilities of boundary component combinations.
\begin{enumerate}
	\item If $\AA(q,i)>0$ (see Figure \ref{fig:triangulation_tube} (a)), then the two boundary components have opposite orientations: $-\beta_{q,i,*} = [v^1_{q,i,*},v^3_{q,i,*},v^2_{q,i,*}]$ and $\alpha_q=[v^1_q,v^2_q, v^3_q]$. We triangulate by matching $v^1_{q,i,*}$ to $v^1_q$, $v^2_{q,i,*}$ to $v^2_q$, and $v^3_{q,i,*}$ to $v^3_q$.
	\item If $\AA(q,i)<0$ (see Figure \ref{fig:triangulation_tube} (b)), then the two boundary components have identical orientations: $-\beta_{q,i,*} = [v^1_{q,i,*},v^3_{q,i,*},v^2_{q,i,*}]$ and $-\alpha_q=[v^1_q,v^3_q, v^2_q]$. We triangulate by matching $v^1_{q,i,*}$ to $v^1_q$, $v^3_{q,i,*}$ to $v^2_q$, and $v^2_{q,i,*}$ to $v^3_q$.
\end{enumerate}
\begin{figure}[!h] 
	\centering 
	\includegraphics[width=\textwidth]{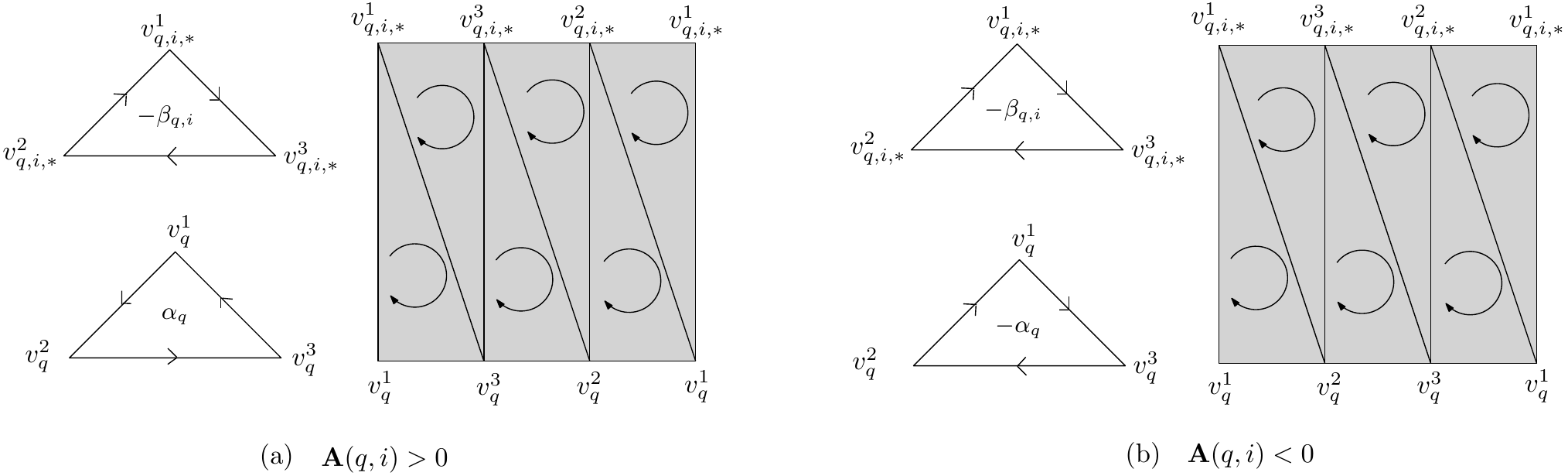} 
	\caption{Oriented triangulation of tubes with opposite or identical boundary orientations.}
	\label{fig:triangulation_tube} 
\end{figure}

In either case, only 6 triangles and 12 edges are required for an oriented triangulation of any tube $\mathcal{T}_{q,i,*}$.
Again, we orient all triangles clockwise.

\subsection{Notations}	
\label{sect:triange_notations}

We introduce several notions and corresponding notations about the constructed 
2-complex. 
These notions and notations will be used in the rest of the paper.

For each $i \in [n]$, let $\calK_i$ be the the union of the triangulated $\calS_i$ and the triangulated tubes that are connected to $\calS_i$, 
which we refer to as the $i$th \emph{complex group};
let $t_i$ be the number of triangles in $\calK_i$;
let $m_i$ be the number of the interior edges in $\calK_i$.

We refer to an edge on a loop $\alpha_q$ as a \emph{boundary edge}
and an edge not on any loop a \emph{interior edge}.
According to our triangulation, each loop $\alpha_q$ has three boundary edges, denoted by $\alpha_q^1=[v^1_q,v^2_q], ~\alpha_q^2=[v^2_q,v^3_q],~ \alpha_q^3=[v^3_q,v^1_q]$. 
A triangle containing a boundary edge is called a \textit{boundary triangle}.
For each boundary edge $\alpha_q^r$, where $q \in [d_1]$ and $r \in [3]$, corresponding to equation 
$\xx(i) - \xx(j) = \bb(q)$, we denote 
the boundary triangles by $\Delta_{q,i,1}^r, \Delta_{q,j,1}^r$
where $\Delta_{q,i,1}^r \in \calT_{q,i},~ \Delta_{q,j,1}^r \in \calT_{q,j}$;
for each boundary edge $\alpha_q^r$, where $q \in [d_1+1:d_2]$ and $r \in [3]$, corresponding to equation 
$\xx(i) + \xx(j) = 2 \xx(k)$,
we denote the boundary triangles by $\Delta_{q,i,1}^r, \Delta_{q,j,1}^r, \Delta_{q,k,1}^r, \Delta_{q,k,2}^r$
where $\Delta_{q,i,1}^r \in \calT_{q,i},~ \Delta_{q,j,1}^r \in \calT_{q,j}, ~\Delta_{q,k,1}^r, \Delta_{q,k,2}^r \in \calT_{q,k}$.

Given any two triangles $\Delta, \Delta' \in \calK$, a \textit{triangle path} from $\Delta$ to $\Delta'$ 
is an ordered collection of triangles $\mathcal{P}=[\Delta^{(0)}=\Delta, \Delta^{(1)},\ldots, \Delta^{(l)}=\Delta']$
such that every neighboring triangles share an edge. The length of $\calP$ is $l$.
A triangle path can also be defined by an ordered collection of edges $\mathcal{P}=[e^{(1)}, \ldots, e^{(l)}]$, where $e^{(i)}$ denotes the edge shared by $\Delta^{(i-1)}$ and $\Delta^{(i)}$, for $i\in[l]$.

\subsection{Algorithm Runtime and Problem Size}

In this section, we show that the reduction algorithm  \textsc{Reduce$\da$To$\b2$}
and the solution mapping algorithm \textsc{MapSoln$\b2$to$\da$}
both run in linear time, 
and \textsc{Reduce$\da$To$\b2$} constructs a 2-complex 
whose size is linear in the number of nonzeros in the input linear equations.

\begin{lemma}[Runtime]
\label{lm:runtime}
Given a difference-average instance \ls~$(\AA,\bb)$ where $\AA\in\R^{d\times n}$, 
Algorithm \textsc{Reduce$\da$To$\b2(\AA,\bb)$} returns $(\partial_2, \gamma, \triangle^c)$ in time $O(\nnz(\AA))$.
Given a solution $\ff$ to \ls~$(\partial_2, \gamma)$, 
Algorithm \textsc{MapSoln$\b2$to$\da(\AA, \bb, \ff, \boldsymbol{\Delta}^c)$}
returns $\xx$ in time $O(n)$.
\end{lemma}

\begin{proof}
	For reduction, \textsc{Reduce$\da$To$\b2(\AA,\bb)$} calls the tube triangulation subroutine for $\norm{\AA}_1$ times, and the punctured sphere triangulation subroutine for $n$ times. The tube triangulation subroutine runs in time $O(1)$ since the there are a constant number of triangles in a tube; and the punctured sphere triangulation subroutine runs in time $O(\norm{\AA(:,j)}_1)$ for the $j$th call, $j\in[n]$. 
	Putting all together, the total runtime of \textsc{Reduce$\da$To$\b2(\AA,\bb)$} is $O\left(\norm{\AA}_1+\sum_{j\in[n]}\norm{\AA(:,j)}_1\right)\leq O(\nnz(\AA))$, where we use the fact $\norm{\AA}_{\max}=2$.
	
	For solution mapping, the runtime of the algorithm \textsc{MapSoln$\b2$to$\da$} is obvious.
\end{proof}

	\begin{lemma}[Size of $\partial_2$]
		\label{lm:size}
		Given a difference-average instance \ls~$(\AA,\bb)$, 
		let $(\partial_2, \gamma, \triangle^c)$ be returned by \textsc{Reduce$\da$To$\b2(\AA,\bb)$}.
		Suppose $\partial_2 \in \mathbb{R}^{m \times t}$. Then,
		\begin{itemize}
			\item $t\leq 22\nnz(\AA)$;
			\item $m\leq 33\nnz(\AA)$;
			\item $\nnz(\partial_2)\leq 66\nnz(\AA)$.			
		\end{itemize}
		\label{lem:size_boundary}
	\end{lemma}
	\begin{proof}
		We first compute the total number of triangles in the constructed 2-complex $\calK$.
		For sphere $\mathcal{S}_j$, we have $\ttil_j=5b_j-4$ triangles by  \eqref{eq:sphere_triangulation}, where $b_j=\sum_{i\in[d]}\abs{\AA(i,j)}$. Therefore, the number of triangles of all spheres is
		\[\sum_{j=1}^n \ttil_j= \sum_{j=1}^n (5\sum_{i\in[d]}\abs{\AA(i,j)} -4)=5\norm{\AA}_1-4n.\]
		Moreover, each boundary component on spheres corresponds to a tube, and each tube has 6 triangles. Hence, the number of triangles of all tubes is $6\norm{\AA}_1$. Putting spheres and tubes together, we get
		\[t=11\norm{\AA}_1-4n\leq 22\nnz(\AA),\]
		where the last inequality is because entries of $\AA$ are bounded by 2.

		Next, we compute the total number of edges in $\calK$. By construction, each triangle has 3 incident edges and each edge is shared by a constant number of triangles (2 for interior edges, and 4 for boundary edges). Thus, we have
		\[m\leq 1.5 t\leq 33\nnz(\AA).\]
		
		Since each column of $\partial_2$ has exactly 3 nonzero entries, we have
		\[\nnz(\partial_2)=3t\leq 66\nnz(\AA).\]
	\end{proof}

%
\subsection{Relation Between Exact Solutions}

\label{sect:exact_solvers}

In this section, we show that the algorithm \textsc{Reduce$\da$To$\b2$} and the algorithm \textsc{MapSoln$\b2$to$\da$} reduce instances
\ls~$(\AA,\bb)$ over $\da$ to instances \ls~over $\b2$,
under the assumption that $\bb$ is in the image of $\AA$.

Given \ls~$(\AA,\bb)$ where $\AA$ is a $d \times n$ matrix in $\da$, let $(\partial_2, \gamma, \triangle^c)$
be returned by \textsc{Reduce$\da$To$\b2$}$(\AA,\bb)$, and let $\calK$ be the 2-complex 
constructed in \textsc{Reduce$\da$To$\b2$}$(\AA,\bb)$
and the boundary operator of $\calK$ is $\partial_2$.
Let $\ff$ be a solution to \ls~$(\AA,\bb)$.

To simplify the analysis, we reorder the columns and rows of $\partial_2$.
The columns $[1:t_1]$ of $\partial_2$ correspond
to the triangles in $\calK_1$, the columns $[t_1+1:t_2]$ correspond
to the triangles in $\calK_2$, and so on. Then, $\ff$ can be written as
\[\ff=\begin{bmatrix}
	\ff_1\\
	\vdots\\
	\ff_n
\end{bmatrix}, \qquad \text{where }  \ff_i\in\R^{t_i}, \forall i\in[n].\]
The rows of $\partial_2$ and the entries of $\gamma$ are:
\begin{align}
\partial_2=\begin{bmatrix}
	\begin{array}{ccc}
		&\BB_1& \\
		&\vdots&\\
		&\BB_d&\\ 
		\hdashline
		\MM_1&&\\
		&\ddots&\\
		&&\MM_n
	\end{array}
\end{bmatrix},\qquad \gamma=\begin{bmatrix}
	\begin{array}{ccc}
		\bb(1)\vecone_3\\
		\vdots\\
		\bb(d)\vecone_3\\ 
		\hdashline
		\mathbf{0}_{m_1}\\
		\vdots\\
		\mathbf{0}_{m_n}
	\end{array}		
\end{bmatrix},
\label{eqn:alpha_decompose}
\end{align}
Here, each submatrix $\BB_q\in\{0,\pm1\}^{3\times t}$ corresponds to the three boundary edges 
$\{\alpha_q^1, \alpha_q^2, \alpha_q^3\}$;
each submatrix $\MM_i\in\{0,\pm1\}^{m_i\times t_i}$ corresponds to all the interior edges in $\calK_i$.
Interior edges in $\calK_i$ and those in $\calK_j$ do not share endpoints if $i \neq j$.
Let $\MM = \diag(\MM_1, \MM_2, \ldots, \MM_n)$.

\begin{claim}
	For each $i \in [n]$, $\ff_i = \alpha \one_{t_i}$ for some $\alpha \in \mathbb{R}$.
	\label{clm:fi}
\end{claim}

\begin{proof}
	For each $i \in [n]$, we have $\MM_i \ff_i = {\bf 0}$.
	This means that for any two triangles $\Delta, \Delta'$ in $\calK_i$ sharing an interior edge, 
	we have $\ff_i(\Delta) = \ff_i(\Delta')$.
	By our construction of $\calK_i$,
	for any two triangles $\Delta, \Delta'$ in $\calK_i$, there exists a triangle path connecting $\Delta$ and $\Delta'$.
	The values of $\ff_i$ at the triangles in this triangle path are equal; in particular, $\ff_i(\Delta) = \ff_i(\Delta')$.
	Thus, the values of $\ff_i$ at all the triangles in $\calK_i$ are equal, that is, $\ff_i = \alpha \one_{t_i}$ for some $\alpha \in \mathbb{R}$.
\end{proof}

\begin{lemma}[Exact solvers in feasible case]
	\label{lm:exact_feasible}
	Given a difference-average instance \ls~$(\AA,\bb)$ where $\bb\in\im(\AA)$, 
	let $(\partial_2, \gamma, \triangle^c)$ be returned by \textsc{Reduce$\da$To$\b2(\AA,\bb)$}, 
	and let $\ff$ be a solution to \ls~$(\partial_2, \gamma)$.
	Then, $\xx\leftarrow$\textsc{MapSoln$\b2$to$\da(\AA, \bb, \ff, \triangle^c)$} is a solution to \ls~$(\AA,\bb)$.
\end{lemma}

\begin{proof}
	By Claim \ref{clm:fi}, we can write $\ff$ as 
	\[
	\ff = \begin{bmatrix}
		\alpha_1 \one_{t_1} \\
		\alpha_2 \one_{t_2} \\
		\vdots \\
		\alpha_n \one_{t_n}
	\end{bmatrix}, ~~~ \text{where} ~ \alpha_1, \ldots, \alpha_n \in \mathbb{R}
	\]
	According to Algorithm \textsc{MapSoln$\b2$to$\da(\AA, \bb, \ff, \triangle^c)$},
	for each $i \in [n]$, $\xx(i) = \alpha_i$.

	Our goal is to show $\AA\xx = \bb$.
	\begin{itemize}
		\item For each difference equation in \ls~$(\AA,\bb)$, say $\xx(i) - \xx(j) = \bb(q)$, we look at 
		the equations in \ls~$(\partial_2, \gamma)$ related to $\BB_q$: 
		\[
			\ff_i(\Delta_{q,i}^r) - \ff_j(\Delta_{q,j}^r) = \bb(q), ~ \forall r \in \{1,2,3\}.
		\]
		thus $\xx(i) - \xx(j) = \bb(q)$ holds.
		\item  For each average equation in \ls~$(\AA,\bb)$, say $\xx(i) + \xx(j) - 2\xx(k) = \bb(q) = 0$, we look at 
		the equations in \ls~$(\partial_2, \gamma)$ related to $\BB_q$:
		\[
			\ff_i(\Delta_{q,i}^r)+\ff_j(\Delta_{q,j}^r)-\ff_k(\Delta_{q,k,1}^r)-\ff_k(\Delta_{q,k,2}^r) 
			= 0,~ \forall r\in\{1,2,3\}.	
		\]
		thus $\xx(i) + \xx(j) - 2\xx(k) = 0$ holds.
	\end{itemize}
\end{proof}


%
\section{Reducing Approximate Solvers for $\da$ to $\b2$ Assuming the Right-Hand Side Vector in the Image of the Coefficient Matrix}
\label{sect:approximate_solvers}

In this section, we show that the algorithm \textsc{Reduce$\da$To$\b2$} and the algorithm \textsc{MapSoln$\b2$to$\da$}
reduce instances
\lsa~$(\AA,\bb,\eps)$ over $\da$ to instances \lsa~$(\partial_2, \gamma, \eps')$ over $\b2$,
under the assumption that $\bb$ is in the image of $\AA$.
Specifically,
we show that by a proper choice of $\eps$, a solution to \lsa~$(\AA,\bb,\eps)$ can be 
converted to a solution to \lsa~$(\partial_2, \gamma, \eps')$ in Section \ref{subsect:error};
we upper bound the condition number of $\partial_2$ in Section \ref{sect:condition_number}.

\subsection{Relation Between Approximate Solutions}
\label{subsect:error}

\begin{lemma}[Approximate solvers in feasible case]
	\label{lm:error_analysis_feasible}
	Given a difference-average instance \lsa~$(\AA,\bb,\epsda)$ where $\bb \in \im(\AA)$, 
	let $(\partial_2, \gamma, \triangle^c)$ be returned by \textsc{Reduce$\da$To$\b2(\AA,\bb)$}.
	Suppose $\ff$ is a solution to \lsa~$(\partial_2, \gamma, \epsb)$
	where 
	\[
		\epsb \le \frac{\epsda}{42\nnz(\AA)},
	\]
	and $\xx$ is returned by \textsc{MapSoln$\b2$to$\da(\AA,\bb,\ff,\boldsymbol{\Delta}^c)$}. 
	Then, $\xx$ is a solution to \lsa~$(\AA,\bb,\epsda)$.
\end{lemma}

\begin{proof}
Since $\bb\in\im(\AA)$, we have 
$\norm{\partial_{2}\ff-\gamma}_{\infty} \le \norm{\partial_{2}\ff-\gamma}_2 \leq \epsb\norm{\gamma}_2$. 

We claim
\begin{align}
	\norm{\AA \xx - \bb}_{\infty}	\le 24 \nnz(\AA)^{1/2} \cdot \epsb \norm{\gamma}_2.
	\label{eqn:axb_infty_norm}
\end{align}
Then, 
$$\norm{\AA \xx - \bb}_{2} \le 24 \nnz(\AA) \cdot \epsb \norm{\gamma}_2.$$
By \textsc{Reduce$\da$To$\b2(\AA,\bb)$}, $\norm{\gamma}_2 \le \sqrt{3} \norm{\bb}_2$.
Thus, 
\[
	\norm{\AA \xx - \bb}_{2} \le 42 \nnz(\AA)^{1/2} \cdot \epsb \norm{\bb}_2
	\le \epsda \norm{\bb}_2,
\]
that is, $\xx$ is a solution to \lsa~$(\AA,\bb,\epsda)$.

To prove Eq. \eqref{eqn:axb_infty_norm},
consider an arbitrary equation in \lsa~$(\AA,\bb,\epsda)$, say $a_i \xx(i) + a_j \xx(j) + a_k \xx(k) = \bb(q)$ where 
$a_i, a_j, a_k \in \{-2,-1,0,1\}$.
According to \textsc{MapSoln$\b2$to$\da(\AA,\bb,\ff,\boldsymbol{\Delta}^c)$},
for each $l \in \{i,j,k\}$, $\xx(l) = \ff(\Delta_l)$ where $\Delta_l \in \calK_l$ is the $l$th central triangle in $\triangle^c$.
Then, 
\[
	a_i \xx(i) + a_j \xx(j) + a_k \xx(k) 
	= a_i \ff(\Delta_i) + a_j \ff(\Delta_j) + a_k \ff(\Delta_k). 
\]
Note that the equation in \lsa~$(\partial_2, \gamma, \epsb)$ 
related to the boundary edge $\alpha_q^1$, shared by triangles $\Delta_{q,i,1}^1, \Delta_{q,j,1}^1$ and 
$\Delta_{q,k,1}^1, \Delta_{q,k,2}^1$ (if the equation is an average equation), satisfies
\[
\abs{a_i \ff(\Delta_{q,i,1}^1) + a_j \ff(\Delta_{q,j,1}^1) + \frac{1}{2} a_k \left( \ff(\Delta_{q,k,1}^1) +  \ff(\Delta_{q,k,2}^1) \right) - \bb(q)}	
\le  	\epsb\norm{\gamma}_2.
\]
For each $\Delta \in \{\Delta_{q,i,1}^1, \Delta_{q,j,1}^1, \Delta_{q,k,1}^1, \Delta_{q,k,2}^1 \}$,
we will replace $\ff(\Delta)$ with its corresponding central triangle $\Delta^c$.
We can find a triangle path connecting $\Delta$ and $\Delta^c$, say $\mathcal{P}=[\Delta^{(0)} = \Delta, \ldots, \Delta^{(l_q)}=\Delta^c]$, 
such that two adjacent triangles $\Delta^{(l)}, \Delta^{(l+1)}$ share an interior edge $e^{(l)}$. 
Then,
\begin{align*}
	\abs{\ff(\Delta)-\ff(\Delta^c)}
	& = \abs{\sum_{l=1}^{l_q} \ff(\Delta^{(l-1)})-\ff(\Delta^{(l)})}  \leq \sum_{l=1}^{l_{q}}  \abs{\ff(\Delta^{(l-1)})-\ff(\Delta^{(l)})} = \sum_{l=1}^{l_{q}} \abs{[\partial_2\ff](e^{(l)})}\\
	& = \sum_{l=1}^{l_{q}} \abs{[\partial_2\ff-\gamma](e^{(l)})} \tag*{since $\gamma(e^{(l)})=0$ for interior edges}\\
	& = \norm{[\partial_2\ff-\gamma] (e^{(1)}:e^{(l_{q})}) }_1 \tag*{where the subvector corresponds to $[e^{(1)},\ldots, e^{(l_{q})}]$}\\
	& \leq \sqrt{t_i}\norm{[\partial_2\ff-\gamma] (e^{(1)}:e^{(l_{q})}) }_2 \tag*{since $l_{q}\leq t_i$}\\
	& \leq \sqrt{t_i} \norm{\partial_2\ff-\gamma}_2 \\
	& \leq \sqrt{t_i} \cdot \epsb  \norm{\gamma}_2 
\end{align*}
Thus, 
\begin{align*}
	& \abs{ a_i \xx(i) + a_j \xx(j) + a_k \xx(k) - \bb(q) } \\
	= & \abs{ a_i \ff(\Delta_i) + a_j \ff(\Delta_j) + a_k \ff(\Delta_k) - \bb(q)} \\
	\le & \underbrace{\abs{a_i \ff(\Delta_{q,i,1}^1) + a_j \ff(\Delta_{q,j,1}^1) + \frac{1}{2} a_k \left( \ff(\Delta_{q,k,1}^1) +  \ff(\Delta_{q,k,2}^1) \right) - \bb(q)}}_{\leq \epsb \norm{\gamma}_2} \\
	& + \underbrace{\abs{\ff(\Delta_{q,i,1}^1) - \ff(\Delta_i)}}_{\le \sqrt{t_i} \epsb \norm{\gamma}_2}
	+ \underbrace{\abs{ \ff(\Delta_{q,j,1}^1) - \ff(\Delta_j)}}_{\le \sqrt{t_j} \epsb \norm{\gamma}_2}
	+ \underbrace{\abs{\ff(\Delta_{q,k,1}^1) - \ff(\Delta_k)}}_{\le \sqrt{t_k} \epsb \norm{\gamma}_2} 
	+ \underbrace{\abs{\ff(\Delta_{q,k,2}^1) - \ff(\Delta_k)}}_{\le \sqrt{t_k} \epsb \norm{\gamma}_2} \\
	\le & 5 \sqrt{t} \cdot \epsb \norm{\gamma}_2 \\
	\le & 24 \nnz(\AA)^{1/2} \cdot \epsb \norm{\gamma}_2
	\tag*{by Lemma \ref{lem:size_boundary}}
\end{align*}
That is, $\norm{\AA\xx - \bb}_{\infty} \le 24 \nnz(\AA)^{1/2} \cdot \epsb \norm{\gamma}_2$.
\end{proof}


\subsection{Bounding the Condition Number of the New Matrix}
\label{sect:condition_number}

In this section, we show that the condition number of $\partial_2$
is upper bounded by a polynomial of $\nnz(\AA), \kappa(\AA)$.

\begin{lemma}[The condition number of $\partial_{2}$]
Given a difference-average instance \lsa  $(\AA,\bb,\epsda)$ where $\bb \in \im(\AA)$, 
let $(\partial_2, \gamma, \triangle^c)$ be returned by \textsc{Reduce$\da$To$\b2(\AA,\bb)$}.
	Then,
	$$\kappa(\partial_{2})\leq 10^9\nnz(\AA)^{9/2}\kappa(\AA)^2.$$
	\label{cor:cond_number_partial2}
\end{lemma}

Note that 
\[
	\kappa^2(\partial_2)=\kappa(\partial_2^\top\partial_2)=\frac{\lambda_{\max}(\partial_2^\top\partial_2)}{\lambda_{\min}(\partial_2^\top\partial_2)}.
\]
We will upper bound $\lambda_{\max}(\partial_2^\top\partial_2)$ and lower bound $\lambda_{\min}(\partial_2^\top\partial_2)$.
Our proof will heavily rely on the Courant-Fischer theorem.

\begin{theorem}[The Courant-Fischer Theorem]
	\label{thm:courant-fischer}
	Let $\MM$ be a symmetric matrix in $\R^{n\times n}$ where $\lambda_{\max}, \lambda_{\min}$ are its maximum and minimum nonzero eigenvalue, respectively. 
	Then
	\[
		\lambda_{\max}=\max_{\xx\neq\bf 0}\frac{\xx^\top\MM\xx}{\xx^\top\xx},
		\qquad \lambda_{\min}=\min_{\xx\perp\Null(\MM),\xx\neq\bf 0}\frac{\xx^\top\MM\xx}{\xx^\top\xx}.
	\]
\end{theorem}

\subsubsection{The Maximum Eigenvalue}

\begin{lemma}[The maximum eigenvalue]
	\label{lm:max_eigenvalue}
	$\lambda_{\max}(\partial_2^\top \partial_2) \le 12$.
\end{lemma}
\begin{proof}
	By the Courant-Fischer theorem,
	\begin{align*}
		\lambda_{\max}(\partial_2^\top \partial_2)
		= \max_{\ff: \norm{\ff}_2=1} \ff^\top \partial_2^\top \partial_2 \ff.
	\end{align*}
	Then for any $\ff$ with $\norm{\ff}_2 = 1$,
	\begin{align*}
		\ff^\top \partial_2^\top \partial_2 \ff
		& = \sum_{i=1}^m \left(\sum_{\Delta: e_i \in \Delta} \partial_2(e_i, \Delta) \ff(\Delta) \right)^2  
		\le 4 \sum_{i=1}^m  \sum_{\Delta: e_i \in \Delta} \ff^2(\Delta)
	 = 12 \norm{\ff}_2^2 = 12.
	\end{align*}
	where the inequality is by the Cauchy-Schwarz inequality.

\end{proof}

\subsubsection{The Minimum Nonzero Eigenvalue}

We start with proving a relation between the null space of $\AA$ and that of $\partial_2$.

\begin{lemma}
	Let 
	\[
		\HH=\begin{bmatrix}
		\vecone_{t_1}&&\\
		&\vdots&\\
		&&\vecone_{t_n}
	\end{bmatrix}\in\R^{t\times n}.
	\]
	Then, $\HH$ is a bijection from $\Null(\AA)$ to $\Null(\partial_2)$.
	\label{lm:1_to_1}
\end{lemma}
\begin{proof}
	Let $\xx \in \Null(\AA)$. By our construction of $\partial_2$, we have $ \HH \xx \in \Null(\partial_2)$.
	For any $\ff \in \Null(\partial_2)$, by the proof of Lemma \ref{lm:exact_feasible}, 
	$\ff = \HH \xx$ for some $\xx \in \mathbb{R}^n$ and $\AA \xx = {\bf 0}$.
\end{proof}

\begin{lemma}[The minimium nonzero eigenvalue]
	\begin{equation}
		\lambda_{\min}(\partial_2^\top\partial_2)
			\geq \frac{\min\{\lambda_{\min}(\AA^\top\AA)^2, 1\}}{10^{16}d^{7}}.
	\end{equation}
	\label{lem:sec5_min_eigenvalue}	
\end{lemma}

\begin{proof}
By the Courant-Fischer theorem,
\[\lambda_{\min}(\partial_2^\top\partial_2)=\min_{\begin{subarray}{c}
		\ff\in\R^t: \ff\perp\Null(\partial_2)\\
		\norm{\ff}_2=1
	\end{subarray}}\ff^\top\partial_2^\top\partial_2\ff.\] 

Let
\[
C = \frac{ \min\{ \lambda_{\min}(\AA^\top\AA),1 \}}{10^{6} d^{2.5}}.
\]
We will exhaust all the vectors in $\{\ff:  \ff \perp \Null(\partial_2), \norm{\ff}_2 = 1\}$
by the following two cases.
\paragraph{Case 1.}
Suppose there exists $i \in [n]$ such that $\calK_i$ contains two triangles $\Delta, \Delta'$ 
satisfying $\abs{\ff(\Delta) - \ff(\Delta')} \ge C$.
Consider a triangle path in $\calK_i$ connecting $\Delta$ and $\Delta'$, 
say $[\Delta^{(0)}=\Delta, \Delta^{(1)}, \ldots, \Delta^{(l)}=\Delta']$
where $l \le 22d$.
There must exists $i^* \in [l]$ such that 
\begin{align*}
	\abs{\ff(\Delta^{(i^*-1)})-\ff(\Delta^{(i^*)})} \ge \frac{C}{l}.
\end{align*}
Note that 
\begin{align*}
	\ff^\top \partial_2^\top  \partial_2 \ff
	\ge \left( \ff(\Delta^{(i^*-1)})-\ff(\Delta^{(i^*)}) \right)^2
	\ge \left(\frac{C}{l}\right)^2
	\ge \frac{ \min\{ \lambda_{\min}(\AA^\top\AA)^2,1 \}}{10^{16} d^{7}}.
\end{align*}

\paragraph{Case 2.} 
Suppose for every $i \in [n]$ and every two $\Delta, \Delta' \in \calK_i$, we have $\abs{\ff(\Delta) - \ff(\Delta')} < C$.
We write $\ff$ as 
\begin{align*}
	\ff = \fftil+\boldsymbol{\epsilon}=\begin{bmatrix}
		\alpha_1 \vecone_{t_1}\\
		\alpha_2 \vecone_{t_2}\\
		\vdots\\
		\alpha_n \vecone_{t_n}
	\end{bmatrix} + \boldsymbol{\epsilon},
\end{align*}
where $\alpha_i$ is the value of $\ff$ at the central triangle of $\calK_i$.
Then, 
\begin{align*}
	& \norm{\boldsymbol{\epsilon}}_2 < \sqrt{t} C \le \sqrt{22d} C = o(1), \\
	& \norm{\fftil}_2 = \norm{\ff-\boldsymbol{\epsilon}}_2 \in \left(\frac{1}{2}, 2 \right).
\end{align*}
We lower bound the quadratic value:
\begin{align*}
	\ff^\top\partial_2^\top\partial_2\ff 
	& = \fftil^\top\partial_2^\top\partial_2\fftil + 2\fftil^\top\partial_2^\top\partial_2\boldsymbol{\epsilon} + \boldsymbol{\epsilon}^\top\partial_2^\top\partial_2\boldsymbol{\epsilon} \\
	& \ge \fftil^\top\partial_2^\top\partial_2\fftil - 2 \norm{\partial_2^\top\partial_2 \fftil}_2 \norm{\boldsymbol{\epsilon}}_2  \\
	& \ge \fftil^\top\partial_2^\top\partial_2\fftil - 4\norm{\partial_2^\top\partial_2}_2\norm{\boldsymbol{\epsilon}}_2  \\
	& \ge \fftil^\top\partial_2^\top\partial_2\fftil - 48 \sqrt{22 d} C
	\tag*{by Lemma \ref{lm:max_eigenvalue}}
\end{align*}
Note that 
\begin{align*}
\fftil^\top\partial_2^\top\partial_2\fftil
& = 3 \norm{\AA \aalpha}_2^2,
\end{align*}
where $\aalpha = (\alpha_1, \ldots, \alpha_n)^\top$.
Write $\aalpha = \aalpha_{\perp} + \aalpha_{0}$, where 
$\aalpha_{\perp}$ is orthogonal to the null space of $\AA$ and $ \aalpha_{0}$ is in the null space of $\AA$.
Then, 
\begin{align}
	\fftil^\top\partial_2^\top\partial_2\fftil
	\ge 3 \lambda_{\min}(\AA^\top \AA) \norm{\aalpha_{\perp}}_2^2.
	\label{eqn:quad_lower_bound}
\end{align}
It remains to lower bound $\norm{\aalpha_{\perp}}_2$.
We can write $\fftil = \HH \aalpha_{\perp} + \HH \aalpha_0$.
By Lemma \ref{lm:1_to_1}, $\HH \aalpha_0 \in \Null(\partial_2)$ and $\HH \aalpha_{\perp} \perp \Null(\partial_2)$.
On the other hand, $\fftil = \ff - \boldsymbol{\epsilon}$ and $\ff \perp \Null(\partial_2)$.
We know 
\[
\norm{\HH \aalpha_{\perp}}_2 \ge \norm{\ff}_2 - \norm{\boldsymbol{\epsilon}}_2,
\]
and thus 
\[
\norm{\aalpha_{\perp}}_2 \ge \frac{1}{\sqrt{t}} \norm{\HH \aalpha_{\perp}}_2
\ge \frac{1}{\sqrt{t}}	- C
> \frac{1}{\sqrt{22 d}}.
\]
Together with Eq. \eqref{eqn:quad_lower_bound},
\[
	\fftil^\top\partial_2^\top\partial_2\fftil
	\ge \frac{3 \lambda_{\min}(\AA^\top \AA)}{22d}.
\]
This completes the proof.
\end{proof}

\begin{proof}[Proof of Lemma~\ref{cor:cond_number_partial2}]
	The proof follows by combining Lemma \ref{lm:max_eigenvalue} and \ref{lem:sec5_min_eigenvalue}.
\end{proof}

%
\section{Reducing Approximate Solvers for $\da$ to $\b2$ in General Case}
\label{sect:regression}

In this section, we show how to reduce \lsa~$(\AA,\bb,\epsda)$ to
instances \lsa~over $\b2$ without requiring the assumption that $\bb$
is in the image of $\AA$ as we did in earlier sections.
In this more general case, \lsa~$(\AA,\bb,\epsda)$ aims to compute an \emph{approximate} solution 
to $\argmin_{\xx} 	\norm{\AA\xx - \bb}_2 = \argmin_{\xx}
\norm{\AA \xx - \PPi_{\AA}\bb}_2$.

We remark that our reduction for this general case does not work for \ls~$(\AA,\bb)$,
which computes an \emph{exact} solution to $\min_{\xx} \norm{\AA\xx - \bb}_2$.
Approximate linear equation solvers, however, can be more interesting in practice, since
numerical errors occur unavoidably during data collection and computation
and approximate solutions may be computed much faster than exact solutions.

\subsection{Warm-Up: Modifying Infeasible Equations While Preserving Solutions}

There is a crucial difference between reductions we use for
\ls~$(\AA,\bb)$ and \lsa~$(\AA,\bb,\epsda)$  when $\bb \in
\im(\AA)$ and in the general case when we may have $\bb \not\in
\im(\AA)$.

To understand this, consider the following feasible system of just two
linear equations in two variables $x$ and $y$.
\begin{align*}
  x - y & = 1 \\
  -x + y & = -1
\end{align*}
A feasible solution is $x = 1$ and $y = 0$. Now, suppose we add
another linear equation that is satisfied by all solutions to the
previous equations, for example, we can simply repeat the first constraint $x - y  = 1$.
It remains true that the existing solutions are feasible.

Now, in contrast, consider an infeasible system of two
linear equations in variables $x$ and $y$.
\begin{align*}
  x - y & = 1 \\
  -x + y & = 0
\end{align*}
This linear equation is \emph{not} feasible. 
In particular, we can consider the associated minimization problem $\argmin_{\xx}
\norm{\AA\xx - \bb}_2$ with 
$\AA = \begin{pmatrix} 1 & -1 \\ -1 & 1 \end{pmatrix}$
and $\bb = \begin{pmatrix} 1 \\ 0 \end{pmatrix}$, for which one
minimizing solution is $\xx = \begin{pmatrix} 1/2 \\ 0 \end{pmatrix}$.
Notice that if we add a row which simply repeats the first
constraint, i.e. $x - y = 1$, then the resulting minimization problem
has 
\[
\AA = \begin{pmatrix} 1 & -1 \\ -1 & 1 \\ 1 & -1  \end{pmatrix}
\text{ and }\bb = \begin{pmatrix} 1 \\ 0 \\ 1 \end{pmatrix}
\] and now $\xx
= \begin{pmatrix} 1/2 \\ 0 \end{pmatrix}$ is no longer a minimizing
solution to $\argmin_{\xx} \norm{\AA\xx - \bb}_2$.
In particular,  $\xx
= \begin{pmatrix} 1/2 \\ 0 \end{pmatrix}$ achieves a value of
$\sqrt{3}/2$, while $\xx'
= \begin{pmatrix} 2/3 \\ 0 \end{pmatrix}$ achives the smaller value
$\sqrt{6}/3$.

However, if we reweigh the first and last row of our system of
inequalities by a factor $1/\sqrt{2}$, so that
\[
\AA = \begin{pmatrix} 1/\sqrt{2} & -1/\sqrt{2} \\ -1 & 1 \\ 1/\sqrt{2} & -1/\sqrt{2}  \end{pmatrix}
\text{ and }\bb = \begin{pmatrix} 1/\sqrt{2} \\ 0 \\ 1/\sqrt{2} \end{pmatrix}
\]
then we in fact maintain that the original solution $\xx
= \begin{pmatrix} 1/2 \\ 0 \end{pmatrix}$ stays optimal.
Thus, when we modify an infeasible system of linear equations,
we have to be very careful about the weight we assign to different
constraints if we are to (approximately) preserve the correspondence
between the optimal solutions of the original and the final problems.
In the following section, we describe a reweighting scheme which can
be combined with our existing reductions to ensure that optimal
solutions are (approximately) preserved, even when the original
problem is infeasible.

\subsection{Reduction Algorithm}

Our reduction is almost the same as the algorithm \textsc{Reduce$\da$To$\b2$},
except that we assign edge weights for the constructed 2-complex.

Suppose we are given an instance \lsa~$(\AA,\bb,\epsda)$ where $\AA$ is a $d \times n$ matrix in $\da$.
Let $(\partial_2, \gamma, \boldsymbol{\Delta}^c) \leftarrow $ \textsc{Reduce$\da$To$\b2$}$(\AA,\bb)$.
Let $\calK$ be the 2-complex whose boundary operator is $\partial_2$.
We compute the edge weights of $\calK$ as follows.

\begin{enumerate}

	\item For each boundary triangle $\Delta^1_{q,i,*}$ where $q \in [d], i \in [n]$ and $* \in \{1,2\}$,
	we find a minimal triangle path $\calP^1_{q,i,*}$ in $\calK_i$ from the central triangle $\Delta_i$ to $\Delta^1_{q,i,*}$.
	Let $E^1_{q,i,*}$ be the set of the edges shared by neighboring triangles in $\calP^1_{q,i,*}$ 
	\footnote{Note that any two neighboring triangles in $\calP^1_{q,i,*}$ share exactly one edge. So, 
	the length of $\calP^1_{q,i,*}$ equals $\abs{E^1_{q,i,*}}$.}.

	\item 
	For each interior edge $e$ and equation $q$, 
	let $k_{q,e}$ be the number of triangle paths indexed by equation $q$ that contain $e$.
	For each equation $q$, let $l_q$ be the sum of the lengths of all the triangle paths indexed by equation $q$.
	Then, we set the weight for edge $e$ to be
	\begin{equation}
		\label{eq:weight}
		w_e=\left\{\begin{matrix}
			1,& \text{if $e$ is a boundary edge}\\ 
			\alpha\sum_{q\in[d]} k_{q,e} l_{q},
			& \text{if $e$ is an interior edge}
		\end{matrix}\right.
	\end{equation}
	where $\alpha = \frac{2}{(\epsda)^2}$.

\end{enumerate}
Let $\WW$ be a diagonal matrix whose diagonals are the edge weights. 
We return $(\WW, \partial_{2}, \gamma, \triangle^c)$.

Note that an edge in $\calK$ may have weight 0.
If we want to make all the edge weights positive, we can 
impose polynomially small edge weights for these 0-weight edges, 
which only affects the error propagation and condition number up to polynomially factors.

We refer to the above algorithm as \textsc{ReduceReg$\da$To$\b2$}.
We use the algorithm \textsc{MapSoln$\b2$to$\da$} to map a solution to \lsa~$(\WW^{1/2}\partial_{2}, \WW^{1/2}\gamma, \epsb)$
back to a solution to \lsa~$(\AA,\bb,\epsda)$, where we choose 
\[
\epsb \le \frac{\epsda}{\sqrt{3 \left(1 + \frac{1}{\alpha} \norm{\bb}_2^2 \nnz(\AA) \norm{\AA}_{\max}^2 \right)}}.	
\]

\paragraph{Computing the Edge Weights in Linear Time.}

Since the weight of an interior edge $e$ only depends on $\calK_i$ that contains edge $e$, 
we will compute the edge weights for the edges in each $\calK_i$ separately.

For each $i \in [n]$, consider a graph $G_i$ whose vertices are the triangles in $\calK_i$ 
and the two vertices are adjacent if and only if the two corresponding triangles share an edge in $\calK_i$.
We run the breadth-first-search (BFS) to construct a shortest-path tree $T_i$ of $G_i$ rooted at $\Delta_i$.
For each boundary triangle $\Delta_{q,i,*}^1$, we choose the triangle path $\calP^1_{q,i,*}$ to be 
the triangle path induced by $T_i$ from the root $\Delta_i$ to the node corresponding to $\Delta_{q,i,*}^1$, whose length is 
the height of the node $\Delta_{q,i,*}^1$.
Since we are only interested in these triangle paths, for simplicity, we remove a node from $T_i$ if it is not a boundary triangle $\Delta_{q,i,*}^1$ for some $q$
and none of its descendant is a boundary triangle $\Delta_{q,i,*}^1$ for some $q$.
After this operation, every leaf of $T_i$ is a boundary triangle $\Delta_{q,i,*}^1$ for some $q$.

Our goal is to count $\sum_{q\in[d]} k_{q,e} l_{q}$ for each edge in the tree $T_i$. 
We observe that an edge $e$ appears in a triangle path from the root $\Delta_i$ to a boundary triangle $\Delta_{q,i,*}^1$
if and only if $\Delta_{q,i,*}^1$ is a descendant of an end-node of $e$ with the higher height.
First, we BFS traverse the tree $T_i$ to store the height of each boundary triangle node.
Then, we traverse the tree $T_i$ to store $l_q$ at each boundary triangle node $\Delta_{q,i,*}^1$ for each $q \in [d]$.
Next, we traverse the tree $T_i$ from the leaf nodes to the root to count 
$\sum_{q\in[d]} k_{q,e} l_{q}$ for each edge $e$. 
The total runtime is linear in the number of triangles in $\calK_j$.

\subsection{Relation Between Exact Solutions}

In this section, we show that by reweighting all the edges in $\calK$,
an exact solution to $(\WW^{1/2}\partial_2, \WW^{1/2}\gamma)$ is close to 
an exact solution to $(\AA,\bb)$, stated in Claim \ref{clm:f_and_x}.
Claim \ref{clm:f_and_x} plays a key role in analyzing approximate solutions and 
condition numbers.

\begin{claim}
	For any $\ff \in \mathbb{R}^t$ and $\xx \leftarrow$ \textsc{MapSoln$\b2$to$\da(\AA,\bb,\ff,\boldsymbol{\Delta}^c)$},
	\[
	\frac{\alpha }{\alpha + 1} \norm{\AA\xx-\bb}_2^2
	\le 
	\norm{\WW^{1/2} \partial_2\ff - \WW^{1/2} \gamma}_2^2.
	\]
	\label{clm:f_and_x}
\end{claim}

Note that if $\ff$ satisfies $\partial_2 \ff = \gamma$, then 
$\norm{\AA\xx-\bb}_2^2
= \norm{ \partial_2\ff - \gamma}_2^2$.
Claim \ref{clm:f_and_x} states that by our choice of weights in $\WW$, 
we can generalize this relation to more general $\ff$.

\begin{proof}

For the convenience of analysis, we construct an auxiliary boundary matrix $\hat{\partial}_2$. 
For each interior edge $e$, let $k_e$ be the number of all the triangle paths containing $e$,
and we split the row $\partial_2(e)$ into $k_e$ copies.
For each copy related to equation $q$, we assign its weight to be $\alpha l_q$.
Let $\hat{\WW}$ be the auxiliary weight matrix, and let $\hat{\gamma}$ be the auxiliary demand vector.
We can check 
\begin{equation}
	\label{eq:reweight}
	\norm{\WW^{1/2} \partial_2\ff - \WW^{1/2}\gamma}_2^2 = 
	\norm{\hat{\WW}^{1/2} \hat{\partial}_2\ff - \hat{\WW}^{1/2}\hat{\gamma}}_2^2.
\end{equation}

We reorder rows of the matrices $\hat{\WW}, \hat{\partial}_2$ and the vector $\hat{\gamma}$ in the following way.
For each $q \in [d]$, let $E_q$ be a (multi)set that is the union of the shared edges in the triangle paths indexed by $q$ 
\footnote{If the $q$th equation is a difference equation, then every edge appears in $E_q$ at most once; 
if the $q$th equation is an average equation, then some edges may appear twice.}.
Then, we reorder rows of $\hat{\WW}, \hat{\partial}_2, \hat{\gamma}$ by grouping those corresponding to the edges in $E_q \cup \{\alpha_q^1 \}$:
\renewcommand{\arraystretch}{1.5}
\begin{equation}
	\label{eq:hat_partial}
	\hat{\partial}_2=\begin{bmatrix}
		\GG_1\\
		\vdots\\
		\GG_d
	\end{bmatrix}, 
	\qquad\text{where}~~ \GG_q=\begin{bmatrix}
	\BB_q\\
	\MM_q
\end{bmatrix}.
\end{equation}
where $\BB_q$ corresponds to the boundary edge $\alpha^1_q$ 
and $\MM_q$ is the submatrix corresponding to all the interior edges in $E_q$.	
Correspondingly, we write
\renewcommand{\arraystretch}{1.5}
\begin{equation}
	\label{eq:hat_w}
	\hat{\WW}=\begin{bmatrix}
		\hat{\WW}_1&&\\
		&\ddots&\\
		&&\hat{\WW}_d
	\end{bmatrix}, 
	~~\text{and}~~ \hat{\gamma}=\begin{bmatrix}
		\hat{\gamma}_1\\
		\vdots\\
		\hat{\gamma}_d
	\end{bmatrix}, \quad\text{where}~~ \hat{\gamma}_q=\begin{bmatrix}
	\bb(q)\\
	\mathbf{0}_{\abs{E_q}}
\end{bmatrix}.
\end{equation}
Let
\[
	\hat{w}_{q,e} 
	\defeq \hat{\WW}_q(e,e)
	= \left\{\begin{matrix}
	1,& \text{if $e$ is a boundary edge}\\ 
	\alpha\cdot l_q, & \text{if $e$ is an interior edge}
\end{matrix}
\right.
\]

For each $q \in [d]$, we define 
\[
	\epsilon_q=\AA(q)\xx-\bb(q), ~~ \xi_q = \BB_q\ff-\bb(q), ~~ 
	\delta_{e}=\MM_q(e)\ff.
\]
We can check 
\begin{align*}
	\eps_q =
	 \one^{\top } \left(
	\GG_q\ff-\hat{\gamma}_q\right)
	= 
	 \xi_q +  \sum_{e \in E_q} \delta_e. 
\end{align*}
By the Cauchy-Schwarz inequality, 
\begin{align*}
	\eps_q^2  = 
	\left(	\xi_q +  \sum_{e \in E_q} \delta_e 
	\right)^2 
	 \le \left(1 + \sum_{e \in E_q} \frac{1}{\hat{w}_{q,e}} \right) 
	\left( \xi_q^2 + \sum_{e \in E_q} \hat{w}_{q,e} \delta_e^2 \right)
	& = \left(1 + \frac{1}{\alpha}\right) 
	\left( \xi_q^2 + \sum_{e \in E_q} \hat{w}_{q,e} \delta_e^2 \right),
\end{align*}
that is, 
\[
\norm{\AA(q) \xx - \bb(q)}_2^2
\le \left(1 + \frac{1}{\alpha}\right)  \norm{\hat{\WW}^{1/2}_q \GG_q \ff - \hat{\WW}^{1/2}_q \hat{\gamma}_q }_2^2.
\]
Summing over all $d$ equations, we get
\[
	\norm{\AA\xx-\bb}_2^2 \le 
	\left(1 + \frac{1}{\alpha}\right)  \norm{\hat{\WW}^{1/2} \hat{\partial}_2\ff - \hat{\WW}^{1/2} \hat{\gamma}}_2^2.
\]
By Eq.  \eqref{eq:reweight}, we get the inequality in the statement.
\end{proof}

Claim \ref{clm:f_and_x} implies the following relation between the exact solutions 
to $(\AA,\bb)$ and those to $(\WW^{1/2} \partial_2, \WW^{1/2} \gamma)$.

\begin{lemma}
	\label{lm:general_case}
	Given any $d \times n$ matrix $\AA \in \da$ and vector $\bb \in \mathbb{R}^d$, 
	let $(\WW, \partial_2, \gamma, \triangle^c) \leftarrow$ 
	\textsc{ReduceReg$\da$To$\b2(\AA,\bb, \epsda)$}. 
	Then, 
	\[
	\frac{\alpha}{\alpha + 1} \min_{\xx}\norm{\AA\xx-\bb}_2^2
	\leq \min_{\ff} \norm{\WW^{1/2} \partial_2\ff- \WW^{1/2} \gamma}_2^2
	\leq \min_{\xx}\norm{\AA\xx-\bb}_2^2.
	\]
\end{lemma}
Remark that Lemma \ref{lm:general_case} can also be stated as 
\[
	\frac{\alpha}{\alpha + 1} \norm{(\II - \PPi_{\AA})\bb}_2^2
	\leq \norm{(\II - \PPi_{\WW^{1/2} \partial_2}) \WW^{1/2} \gamma}_2^2
	\leq \norm{(\II - \PPi_{\AA})\bb}_2^2.
\]
We do not prove equalities. But as $\alpha \rightarrow \infty$,  the leftmost hand side 
and the rightmost side hand are equal.

\begin{proof}

Let $\ff^* \in \argmin_{\ff} \norm{\WW^{1/2}\partial_2\ff -\WW^{1/2} \gamma}_2$
and $\xx^* \in \argmin_{\xx} \norm{\AA \xx - \bb}_2$.
Let $\xx \leftarrow$ \textsc{MapSoln$\b2$to$\da(\AA,\bb,\WW^{1/2}\partial_2,\ff^*)$}.
By Claim \ref{clm:f_and_x}
\[
\frac{\alpha}{\alpha+1}	\norm{\AA\xx^*-\bb}_2^2 \le	\frac{\alpha}{\alpha+1} \norm{\AA\xx-\bb}_2^2 
\le \norm{\WW^{1/2} \partial_2 \ff^* - \WW^{1/2} \gamma}_2^2,
\]
which is the first inequality in the lemma statement.
Let 
$$\ff = \begin{bmatrix}
	\xx^*(1) \one_{t_1} \\
	\vdots \\
	\xx^*(n) \one_{t_n}
\end{bmatrix}.$$
Then, 
\[
	\norm{\WW^{1/2} \partial_2 \ff^* - \WW^{1/2} \gamma}_2^2
	\le \norm{\WW^{1/2} \partial_2 \ff - \WW^{1/2} \gamma}_2^2 =  \norm{\AA \xx^* - \bb}_2^2,
\]
which is the second inequality in the lemma statement.
\end{proof}

\subsection{Relation Between Approximate Solutions}

Linear equation problem \ls~$(\AA,\bb)$ aims to find a vector $\xx$ 
such that $\AA\xx = \PPi_{\AA} \bb$.
In our setting, both $\AA$ and $\bb$ have integer entries. 
We will need the following claim to lower bound $\AA\xx = \PPi_{\AA} \bb$.

\begin{claim}
	\label{lm:projection_norm}
	Let $\AA \in \mathbb{R}^{d \times n}$ and $\bb \in \mathbb{R}^d$ 
	such that $\norm{\AA^\top\bb}_2 \ge 1$.
	Then,
	\[
	\norm{\PPi_{\AA}\bb}_2^2 
	\geq \frac{1}{\lambda_{\max}\left( \AA^{\trp} \AA \right)}.
	\]
\end{claim}

\begin{proof}
Note that 
\begin{align*}
	\norm{\PPi_{\AA}\bb}_2^2 & = \bb^\top \AA (\AA^\top \AA)^{\dagger} \AA^\top \bb 
	\ge \lambda_{\min}((\AA^\top \AA)^{\dagger}) \norm{\AA^\top \bb}_2^2
	\ge \frac{1}{\lambda_{\max}\left( \AA^{\trp} \AA \right)}.
\end{align*}
\end{proof}

Claim \ref{lm:projection_norm} and Lemma \ref{lm:general_case} enable us to relate $\norm{\PPi_{\AA}\bb}_2$
with $\norm{\PPi_{\WW^{1/2} \partial_2} \WW^{1/2}\gamma}_2$.

\begin{claim}
	\label{lm:norm_ratio}
	\[
		\norm{\PPi_{\WW^{1/2} \partial_2} \WW^{1/2} \gamma}_2^2
		\le \left(1 + \frac{1}{\alpha} \lambda_{\max} (\AA^\top \AA)  \norm{\bb}_2^2 \right)
		\norm{\PPi_{\AA} \bb}_2^2.
	\]
\end{claim}

\begin{proof}
	We apply Lemma~\ref{lm:general_case}.
	\begin{align*}
		\norm{\PPi_{\AA}\bb}_2^2 &= \norm{\bb}_2^2 - \norm{(\II-\PPi_{\AA})\bb}_2^2\\
		&\geq \norm{\WW^{1/2} \gamma}_2^2 - \frac{\alpha+1}{\alpha}\norm{(\II-\PPi_{\WW^{1/2} \partial_2}) \WW^{1/2}\gamma}_2^2 
		\\
		& = \norm{\PPi_{\WW^{1/2}  \partial_2} \WW^{1/2}\gamma}_2^2
		- \frac{1}{\alpha}\norm{(\II-\PPi_{\WW^{1/2}  \partial_2}) \WW^{1/2} \gamma}_2^2\\
		&\geq \norm{\PPi_{\WW^{1/2} \partial_2} \WW^{1/2}\gamma}_2^2 - \frac{1}{\alpha}\norm{(\II-\PPi_{\AA})\bb}_2^2 
	\end{align*}
	Since $\AA, \bb$ have integer entries, by Claim \ref{lm:projection_norm}, 
	\begin{align*}
		\norm{(\II - \PPi_{\AA})\bb}_2^2 \le \norm{\bb}_2^2 
		\le \norm{\bb}_2^2 \cdot \lambda_{\max}(\AA^\top\AA) \norm{\PPi_{\AA}\bb}_2^2. 
	\end{align*}
	Thus, 
	\[
		\norm{\PPi_{\WW^{1/2} \partial_2} \WW^{1/2} \gamma}_2^2
		\le \left(1 + \frac{1}{\alpha} \lambda_{\max} (\AA^\top \AA)  \norm{\bb}_2^2 \right)
		\norm{\PPi_{\AA} \bb}_2^2.
	\]
\end{proof}

Now, we apply Claim \ref{clm:f_and_x} and Lemma \ref{lm:norm_ratio} to prove 
a relation between approximate solutions.

\begin{lemma}[Approximate solvers in general case]
	Given a difference-average instance \lsa~$(\AA,\bb,\epsda)$, 
	let $(\WW, \partial_2, \gamma,\triangle^c) \leftarrow$ \textsc{ReduceReg$\da$To$\b2(\AA,\bb,\alpha)$}. 
	Suppose $\ff$ is a solution to \lsa~$(\WW^{1/2}\partial_2, \WW^{1/2}\gamma, \epsb)$, where 
	$\epsb \le \frac{\epsda}{10}$.
	Let $\xx \leftarrow$ \textsc{MapSoln$\b2$to$\da(\AA,\bb,\ff,\triangle^c)$}.
	Then, $\xx$ is a solution to \lsa~$(\AA,\bb,\epsda)$.
	\label{lem:regressions_approx}
\end{lemma}

\begin{proof}
	By Claim~\ref{clm:f_and_x}, we have 
	\begin{align}
		\norm{\AA\xx - \bb}_2^2 \le \frac{\alpha + 1}{\alpha} \norm{\WW^{1/2} \partial_2 \ff - \WW^{1/2} \gamma}_2^2.	
		\label{eqn:fx}	
	\end{align}
	Also note that 
	\begin{align*}
		 \norm{\AA\xx - \bb}_2^2 &= \norm{\AA\xx - \PPi_{\AA}\bb}_2^2 + \norm{(\II - \PPi_{\AA})\bb}_2^2, \\
		 \norm{\WW^{1/2} \partial_2 \ff - \WW^{1/2} \gamma}_2^2	&= \norm{\WW^{1/2} \partial_2 \ff - \PPi_{\WW^{1/2}  \partial_2} \WW^{1/2} \gamma}_2^2	
		+ \norm{(\II - \PPi_{\WW^{1/2}\partial_2}) \WW^{1/2} \gamma}_2^2.
	\end{align*}	
	Plugging these into Eq. \eqref{eqn:fx} and apply Lemma \ref{lm:general_case},
	\begin{align*}
		\norm{\AA\xx - \PPi_{\AA}\bb}_2^2
		& \le \frac{\alpha + 1}{\alpha} \left(
		\norm{\WW^{1/2} \partial_2 \ff - \PPi_{\WW^{1/2} \partial_2} \WW^{1/2} \gamma}_2^2	
		+ \norm{(\II - \PPi_{\WW^{1/2}\partial_2})\WW^{1/2} \gamma}_2^2
		\right) - \norm{(\II - \PPi_{\AA})\bb}_2^2 \\
		& \le \frac{\alpha + 1}{\alpha}
		\norm{\WW^{1/2} \partial_2 \ff - \PPi_{\WW^{1/2} \partial_2} \WW^{1/2} \gamma}_2^2
		+ \frac{1}{\alpha} \norm{(\II - \PPi_{\AA})\bb}_2^2. 
	\end{align*}
	By the fact that $\ff$ is a solution to \lsa~$(\WW^{1/2} \partial_2,\WW^{1/2} \gamma, \epsb)$ 
	and by Claim \ref{lm:norm_ratio},
	\begin{align*}
		  \norm{\WW^{1/2} \partial_2 \ff - \PPi_{\WW^{1/2} \partial_2} \WW^{1/2} \gamma}_2^2
		& \le \left(\epsb\right)^2 \norm{\PPi_{\WW^{1/2} \partial_2} \WW^{1/2} \gamma}_2^2 
		\le \frac{\left(\epsda\right)^2}{3} \norm{\PPi_{\AA}\bb}_2^2.
	\end{align*}
In addition, 
	\begin{align*}
		\frac{1}{\alpha} \norm{(\II - \PPi_{\AA})\bb}_2^2 
		\le \frac{1}{\alpha} \norm{\bb}_2^2  
		 \le \frac{\left(\epsda\right)^2}{2} \norm{\PPi_{\AA}\bb}_2^2. 
	\end{align*}
Thus, 
\[
\norm{\AA\xx - \PPi_{\AA}\bb}_2^2 \le \left(\epsda\right)^2 \norm{\PPi_{\AA}\bb}_2^2,
\]
that is, $\xx$ is a solution to \lsa~$(\AA,\bb,\epsda)$.
\end{proof}

\subsection{Bounding the Condition Number of the New Matrix}

We will upper bound the maximum eigenvalue of $\partial_2^\top \WW \partial_2$
and lower bound its minimum nonzero eigenvalue.
The proofs are similar to those in Section \ref{sect:condition_number}, 
which bound the eigenvalues of $\partial_2^\top \partial_2$.

\begin{lemma}
	$\kappa(\partial_2^\top \WW \partial_2) = O \left( (\epsda)^{-2} \nnz(\AA)^{15/2} \kappa(\AA^\top \AA)  \right)$.
\end{lemma}

\begin{proof}
	The proof follows the same proof line in Section \ref{sect:condition_number} 
	for $\WW = \II$. 
	Here, we lose a factor $w_{\max}$ when we upper bound $\lambda_{\max}(\partial_2^\top \WW \partial_2)$,
	and we lose a factor $\frac{w_{\max}}{w_{\min}}$ when we lower bound $\lambda_{\max}(\partial_2^\top \WW \partial_2)$,
	where $w_{\max}$ is the maximum diagonal in $\WW$ 
	and $w_{\min}$ is the minimum nonzero diagonal in $\WW$.
	By our setting, $w_{\max} = O(\alpha \nnz(\AA)^2)$ and $w_{\min} = \alpha$, 
	where $\alpha = 2 (\epsda)^{-2}$.
\end{proof}


\printbibliography

\appendix

\section{Reducing General Linear Equations to Difference-Average Linear Equations}
\label{sect:appendix_first_reduction}


\newcommand{\matclass}{\calM}
\newcommand{\algo}{\calA}

\newcommand{\GLD}{\textsc{lsd}}
\newcommand{\LSD}{\textsc{lsd}}

\newcommand{\MCLA}{\textsc{mcla}}
\newcommand{\MAP}{\textsc{map}}

\newcommand{\YES}{\textsc{yes}}
\newcommand{\NO}{\textsc{no}}

\newcommand{\asoln}{approximate solution}

\newcommand{\genCl}{\calG}
\newcommand{\genZCl}{\calG_{\text{z}}}
\newcommand{\genZtwoCl}{\calG_{\text{z},2}}
\newcommand{\mctwoCl}{\calM\calC_{2}}
\newcommand{\mctwostrictCl}{{\calM\calC_{2}^{>0}}}
\newcommand{\mctwostrictintCl}{\calM\calC_{{2,\bz}}^{>0}}
\newcommand{\trusstwoCl}{\calT_{2}}
\newcommand{\tvtwoCl}{\calV_{2}}

\newcommand{\algGZTtoMCT}{\ensuremath{\textsc{Reduce\,}\genZtwoCl\textsc{to}\da}}
\newcommand{\algMCTtoGZTsoln}{\ensuremath{\textsc{MapSoln\,}\da\textsc{to}\genZtwoCl}}
\newcommand{\algMCToGZGadget}{\ensuremath{\mctwoCl\textsc{Gadget}}}
\newcommand{\xxa}{\xx^{\textnormal{A}}}
\newcommand{\xxb}{\xx^{\textnormal{B}}}
\newcommand{\xxaux}{\xx^{\textnormal{aux}}}
\newcommand{\cca}{\cc^{\textnormal{A}}}
\newcommand{\ccb}{\cc^{\textnormal{B}}}
\newcommand{\epsa}{\eps^{\textnormal{A}}}
\newcommand{\xxbopt}{\xx^{\textnormal{B}*}}
\newcommand{\xxaopt}{\xx^{\textnormal{A}*}}
\newcommand{\xxatil}{\widetilde{\xx}^{\textnormal{A}}}
\newcommand{\xxauxopt}{\xx^{\textnormal{aux}*}}
\newcommand{\nb}{n_B}
\newcommand{\na}{n_A}
\newcommand{\vecppi}{\pp^i}
\newcommand{\qqj}{\qq^j}
\newcommand{\ppxa}{\pp(\xxa)}
\newcommand{\yyb}{\yy^{\textnormal{B}}}
\newcommand{\yya}{\yy^{\textnormal{A}}}
\newcommand{\yyaux}{\yy^{\textnormal{aux}}}
\newcommand{\zzb}{\zz^{\textnormal{B}}}
\newcommand{\zza}{\zz^{\textnormal{A}}}
\newcommand{\zzaux}{\zz^{\textnormal{aux}}}

In this section, we prove Theorem \ref{thm:prelim_kz17}: 
\lsa~over $\da$ is sparse-linear-equation complete.

We show a nearly-linear time reduction from linear equation problems
over matrices in $\calG$ to linear equation problems over matrices in $\da$.
This reduction is implicit in Section 4 of~\cite{KZ17}, as an intermediate step 
to reduce linear equation problems over matrices in $\calG$ to matrices in a slight generalization of Laplacians.
We explicitly separate the reduction step and simplify the proofs in \cite{KZ17},
which might be of independent interest.

Recall that a matrix in $\calG$ has polynomially bounded integer entries and polynomially bounded condition numbers, 
and they do not have all-0 rows or all-0 columns;
a matrix $\AA \in \da$ only has two types of rows such that 
if we multiply $\AA$ to a vector $\xx$, then 
the entries of $\AA\xx$ are in the form of either $\xx(i) - \xx(j)$ or $\xx(i) + \xx(j) - 2\xx(k)$.

The reduction in \cite{KZ17} has three steps: 
\begin{enumerate}
	\item Reduce linear equation problems over matrices in $\calG$ to  matrices in $\calG_z$, 
	a subset of $\calG$ containing matrices with row sum 0. 
	Given an instance $(\GG, \bb)$
	where $\GG \in \calG$, we construct a new instance  $(\GG', \bb)$
	where $\GG' = \begin{bmatrix}
		\GG& -\GG \one
	\end{bmatrix} \in \calG_z$.
	\item Reduce linear equation problems over matrices in $\calG_z$ to  matrices in $\calG_{z,2}$, 
	a subset of $\calG_{z}$ containing matrices such that the sum of positive entries in each row 
	is a power of 2.
	Given an instance  $(\GG', \bb)$
	where $\GG' \in \calG_{z}$, we construct a new  instance  $(\GG'', \bb'')$
	where $\GG'' = \begin{bmatrix}
		\GG' & \gg & -\gg \\
		\zero & 1 & -1
	\end{bmatrix}
 \in \calG_{z,2}$ 
	and $\bb'' = \begin{bmatrix}
			\bb \\
		0
	\end{bmatrix}$.
	\item Reduce linear equation problems over matrices in $\calG_{z,2}$ to  matrices in $\da$.
\end{enumerate}

The first two steps are proved in Section 7 of~\cite{KZ17}, by standard tricks. 
We will focus on the third step.

In the rest part of this section, to be consistent with the notations used in~\cite{KZ17}, 
we use subscripts to denote entries of a matrix or a vector. 
Let $\AA_i$ denote the $i$th row of a matrix $\AA$, and $\AA_{ij}$ denote the $(i,j)$th entry of $\AA$. 
Let $\xx_i$ denote the $i$th entry of a vector $\xx$, 
and $\xx_{i:j}$ denote the subvector of entries $\xx_i,\xx_{i+1},\ldots, \xx_j$. 
Moreover, we let $\norm{\AA}_{\max}$ be the maximum magnitude of the entries of $\AA$.

\subsection{Reduction Algorithm}

Given an instance of linear equation problems over $\calG_{z,2}$, our construction of
an instance over $\da$ does not depend on the error parameter.
Therefore, we will describe our construction without the error parameters; then we will 
explain how to set the error parameter for the construction instance over $\da$.

Let $(\AA,\cca)$ be an instance over $\calG_{z,2}$.
The idea is to write each equation in $(\AA, \cca)$ as a set of difference equations and average equations, 
via bitwise decomposition.
Same as \cite{KZ17},
we first explain the idea by an example: 
\begin{align}
	3\xx_1 + 5\xx_2 -  \xx_3 - 7\xx_4 = 1.
	\label{eqn:da_example}
\end{align}
We will manipulate the variables with positive coefficients and the variables with negative coefficients separately.
Let us take the positive coefficients as an example.
We pair all the variables with the odd positive coefficients and replace each pair with a new variable via an average equation. 
In this example, we pair $\xx_1, \xx_2$;
we then create a new variable $\xx_{12}$ and a new average equation 
$$\xx_1 + \xx_2 - 2\xx_{12} = 0.$$
Plugging this into Eq. \eqref{eqn:da_example}, we get 
\begin{align*}
	2 \xx_1 + 4 \xx_2 + 2\xx_{12} -  \xx_3 - 7\xx_4 = 1.
\end{align*}
We pull out a factor 2:
\begin{align*}
	2 (\xx_1 + 2 \xx_2 + \xx_{12}) -  \xx_3 - 7\xx_4 = 1.
\end{align*}
We repeat the pair-and-replace process to pair $\xx_1, \xx_{12}$
with a new variables $\xx_{1,12}$ and a new average equation $\xx_1 + \xx_{12} - 2\xx_{1,12} = 0$:
\begin{align*}
	2 (2\xx_{1,12} + 2 \xx_2 ) -  \xx_3 - 7\xx_4 = 1.
\end{align*}
Pull out a factor 2:
\begin{align*}
	4 (\xx_{1,12} +  \xx_2 ) -  \xx_3 - 7\xx_4 = 1.
\end{align*}
Repeat pair-and-replace with $\xx_{1,12} +  \xx_2 - 2\xx_{1,12,2} = 0$:
\begin{align*}
	8\xx_{1,12,2}  -  \xx_3 - 7\xx_4 = 1.
\end{align*}
Similarly, we can use a sequence of average equations for the 
variables with odd coefficients, and the above equation becomes:
\begin{align*}
	8\xx_{1,12,2}  -  8 \xx_{34} = 1.
\end{align*}
where $\xx_{34}$ is a new variable.
The final equation is a difference equation.

The above reduction relies on that the sum of all the coefficients 
is 0 and the sum of all the positive coefficients are a power of 2.

Pseudo-code of the reduction from linear equation problems over $\calG_{z,2}$
to $\da$ is shown in Algorithm~\ref{alg:GZ2toMC2}.
Algorithm~\ref{alg:GZ2toMC2solnback} shows how to map a solution to the instance over $\da$
back to a solution to the original instance over $\calG_{z,2}$.

\begin{algorithm}[!h]	
	\caption{\label{alg:GZ2toMC2}\algGZTtoMCT}
	\KwIn{$(\AA,\cca)$ where $\AA \in \genZtwoCl$ is an $m \times
	n$ matrix, $\cca \in \rea^n$, and 
	$\alpha > 0$.}
	\KwOut{
		$(\BB,\ccb)$ where $\BB \in \da$ is an $m \times
		n$ matrix, and $\ccb \in \rea^n$.
		}
    
	  $\mathcal{X} \assign \{\uu_1, \ldots, \uu_n\}$ \tcp*{$\da$ variables and index of new variables}
    
	Let $\xx$ be the vector of variables corresponding to the set of variables $\mathcal{X}$
    
	$t \assign n+1$
    
	$\mathcal{A} \assign \emptyset$, $\mathcal{B} \assign \emptyset$
    \tcp*{Multisets of main and $\da$ auxiliary equations respectively}
    
	\For{each equation $1 \leq i \leq m$ in $\AA$} 
	{
		\If{$\AA_i$ is already a difference equation or an average equation}
		{
			$\mathcal{A} \assign \mathcal{A} \cup
			\{\AA_{i j_+} \uu_{j_+} - \AA_{i j_+} \uu_{j_-} =
			\cc_i\}$

			$w_i \assign \frac{\alpha}{\alpha + 1}$
		}
		\Else
		{
			Let $\mathcal{I}^{+1} \assign \{j: \AA_{ij} > 0\}, \mathcal{I}^{-1} \assign \{j: \AA_{ij} < 0 \}$
    	\label{lin:mcAlready}
		
			$\mathcal{A}_{i} \assign  \setof{\AA_{i} \uu = \cc_{i}}$ \label{lin:mainConstraint}

			$\mathcal{B}_{i} \assign \emptyset $
			
			\For{$s = -1,+1$}
			{
				$r \assign 0$
			
				\While{$\mathcal{A}_i$ is neither a difference equation nor an average equation 
				\label{lin:while_start}
				\label{lin:varReplace}} 
				{
					For each $j$, let $\AAhat_{ij}$ be the coefficient of $\uu_j$ in $\mathcal{A}_{i}$.
					
					$\mathcal{I}_{odd}^{s} \assign \{j \in \mathcal{I}^{s}: \lfloor |\AAhat_{ij}|  / 2^{r} \rfloor \mbox{ is odd} \}$
					
					Pair the indices of $\mathcal{I}_{odd}^{s}$ into disjoint pairs $(j_{1},l_{1}), (j_{2},l_{2}), \ldots$
					
					\For{each pair of indices $(j_{k},l_{k})$}
					{
						$\mathcal{A}_{i} \leftarrow \mathcal{A}_{i}+s\cdot 2^{r}
						\setof{ \left( 2 \uu_{t} - \left( \uu_{j_{k}} + \uu_{l_{k}}
						\right)\right) = 0}$  \footnotemark

						$\mathcal{B}_{i} \leftarrow \mathcal{B}_{i} \cup  s\cdot 2^{r} 
						\setof{ \left( 2 \uu_{t} - \left( \uu_{j_{k}} + \uu_{l_{k}}
						\right)\right) = 0}$   \label{lin:updateBi}

						$\mathcal{X} \assign \mathcal{X}
						\cup \{\uu_t\}$, update $\xx$ accordingly

						$t \assign t+1$

						$r \assign r+1$
					}
				} \label{lin:while_end}
			}

			$w_{i} \assign \alpha \sizeof{\mathcal{B}_{i} }$

			$\mathcal{B} \assign
			\mathcal{B} \cup w_{i}^{1/2}  \cdot \mathcal{B}_{i}$  
			\label{lin:weight}

			$\mathcal{A} \assign \mathcal{A}
			\cup \mathcal{A}_{i}$.
		}
		}    
    \Return {$\BB, \cc$ s.t. $\BB\xx = \cc$ corresponds to the equations in $\mathcal{A} \cup \mathcal{B}$,
    on the variable set $\mathcal{X}$.}
\end{algorithm}
\footnotetext{{Note that given two singleton multi-sets each containing a single equation, e.g.
		$\setof{\aa_{1}^{\trp} \xx = c_1}$ and $\setof{\aa_{2}^{\trp} \xx  =
			c_2}$ where $\aa_1, \aa_2$ are vectors, we define
		$\setof{\aa_{1}^{\trp} \xx = c_1} +\setof{\aa_{2}^{\trp} \xx  = c_2} 
		=
		\setof{\aa_{1}^{\trp} \xx +  \aa_{2}^{\trp} \xx  =c_1+ c_2}$
		and we define 
		$\setof{\aa_{1}^{\trp} \xx = c_1} \union \setof{\aa_{2}^{\trp} \xx  = c_2} 
		=
		\setof{\aa_{1}^{\trp} \xx = c_1, \aa_{2}^{\trp} \xx  = c_2}$.}}

\paragraph{Notations.} 

We will follow the notations in~\cite{KZ17}.
The central object created by Algorithm~\ref{alg:GZ2toMC2} is the matrix
$\BB$, which contains both original and new equations and variables.
We will superscript the variables with $^{\textnormal{A}}$ to
distinguish variables appear in the original equation $\AA \xxa =
\cca$ from new variables.
We will term the new variables as $\xxaux$, and write a vector
solution to the new problem, $\xxb$, as:
\begin{equation}
\xxb = \begin{bmatrix}
	\xxa \\ 
	\xxaux
\end{bmatrix}.
\label{eq:MC2BDecomposeX}
\end{equation}
Let $\na$ be the dimension of $\xxa$, and $\nb$ be the dimension of $\xxb$, respectively.

Furthermore, we will distinguish the equations in $\BB$ into ones
formed from manipulating $\AA$, i.e. the equations added to
the set  $\mathcal{A}$, from the auxiliary equations, i.e. the
equations added to the set $\mathcal{B}$.
We use $\WW^{1/2} = \diag(w_{i}^{1/2}) $ to refer to the diagonal matrix of
weights $w_{i}$ applied to the auxiliary equations $\mathcal{B}$ in
Algorithm~\ref{alg:GZ2toMC2}.
In Algorithm~\ref{alg:GZ2toMC2}, a real value $\alpha > 0$ is set
initially and used when computing the weights $w_{i}^{1/2}$.
For convenience, throughout most of this section, we will treat
$\alpha$ as an arbitrary constant, and only eventually substitute in
its value to complete our main proof.
This leads to the following representation of $\BB$
and $\ccb$ which we will use throughout our analysis of the algorithm:
\begin{equation}
\BB
= \begin{bmatrix}
	\AAhat \\
	\WW^{1/2} \BBhat
\end{bmatrix}.
\label{eq:MC2BDecomposeB}
\end{equation}
Here the equations of $\AAhat$
corresponds to $\mathcal{A}$ in \algGZTtoMCT, and
$\BBhat$ corresponds to the auxiliary constraints,
i.e. equations of $\mathcal{B}$ in \algGZTtoMCT.
Also, the vector $\ccb$ created is simply an extension of $\cca$:
\begin{equation}
\ccb =\begin{bmatrix}
	\cca\\ 
	\zero
\end{bmatrix}.
\label{eq:MC2BDecomposeC}
\end{equation}

Finally, as Algorithm~\ref{alg:GZ2toMC2} creates new equations for
each row of $\AA$ independently, we will use 
$S_i$ to denote the subset of indices of the rows of $\BBhat$
that is created from $\AA_{i}$, and denote $m_i \defeq \abs{S_i}$,
and use $\BBhat_{i}$
 to denote these part of $\BBhat$ that corresponds
to these rows.

We observe that the sets of auxiliary variables created for $\BBhat_i$'s 
are disjoint.

\begin{algorithm}[H]
    \caption{\label{alg:GZ2toMC2solnback}\algMCTtoGZTsoln} 
    \KwIn{$m \times n$ matrix $\AA \in \genZtwoCl$, 
    $m' \times n'$ matrix $\BB \in \da$,
    vector $\cca \in \rea^{m}$,
    vector $\xxb \in \rea^{n'}$. } 
    \KwOut{ Vector $\xxa \in \rea^{n}$.} 
    \If{$\AA^{\trp} \cca = {\bf 0}$} 
    	{\Return {$\xxa \assign {\bf 0}$}}
    \Else
    	{\Return {$\xxa \assign \xxb_{1:n}$	}}
\end{algorithm}

\begin{lemma}
\label{lem:GZToMC2NNZ}
Let $\AA \in \calG_{z,2}$, and let $\BB$ be returned by running Algorithm \algGZTtoMCT ~on $\AA$.
Then, 
\begin{itemize}
\item $\nnz(\BB)= O(  \nnz(\AA) \log \norm{\AA}_{\max}) $;
\item  both the dimensions of $\BB$ are $ O(\nnz(\AA) \log
  \norm{\AA}_{\max})$.
\end{itemize}
\end{lemma}

\begin{proof}[Proof sketch.]
The second statement is implied by the first one.

	To prove the first statement, we focus on Algorithm \algGZTtoMCT~ for a single equation in $\AA_i$.
	The number of the while-iterations from line \ref{lin:while_start} to  \ref{lin:while_end}
	is at most $\log \norm{\AA}_{\max}$.
	We observe that 
	(1) in each iteration, the number of newly created variables and equations is at most
	the value of $\nnz(\calA_i)$ at the beginning of the iteration;
	(2) once an auxiliary variable is added to $\calA_i$, it must be replaced with a new auxiliary variable in the next iteration (if exists).
	Therefore, $\nnz(\calA_i) = O(\nnz(\AA_i))$ for every iteration.
So, $\nnz(\BB) = O(\nnz(\AA) \log \norm{\AA}_{\max})$.
\end{proof}

\subsection{Relation Between Exact Solvers}
\label{subsec:GZToMCExactSolver}

The following lemma characterizes the relation between the error of a solution to $(\AA,\cca)$
to that to $(\BB,\ccb)$.
We will exploit it to derive relations between exactly and approximately solving 
$(\AA,\cca)$ and $(\BB,\ccb)$, and to establish a bound for the condition number of the new 
constructed matrix $\BB$.

\begin{lemma}
\label{clm:exactReduction}
For any fixed $\xxa \in \rea^{n}$,
\[
\norm{\AA\xxa - \cca}_2^2
=
\frac{\alpha+1}{\alpha} \min_{\xxaux} 
\norm{\BB \begin{bmatrix}
		\xxa\\ 
		\xxaux
	\end{bmatrix}  - \ccb}_2^2.
\]
\end{lemma}

Lemma~\ref{clm:exactReduction} immediately implies the following corollaries.

\begin{corollary}
  Given an \ls~$(\AA,\cca)$ where $\AA \in
\genZtwoCl$, let \ls~$(\BB,\cc) \leftarrow \algGZTtoMCT(\AA,\cca,0)$.
Then, $\BB \in \da$, and if $\xxb =\begin{bmatrix}
	\xxa\\ 
	\xxaux
\end{bmatrix}$ is a solution to \ls~$(\BB,\ccb)$, then $\xxa$ is a solution to \ls~$(\AA,\cca)$.
\end{corollary}

\begin{corollary}
\[
\min_{\xxa}	\norm{\AA\xxa - \cca}_2^2
= \frac{\alpha+1}{\alpha} \min_{\xxb} 
\norm{\BB \xxb  - \ccb}_2^2.
\]
\end{corollary}

The optimal solutions to $\min_{\xxb} \norm{\BB\xxb -
  \ccb}_2$ and $\min_{\xxa} \norm{\AA\xxa - \cca}_2$ have a
one-to-one map. 
However, the optimal values are different; when $\alpha \rightarrow \infty$, the two optimal values approach
the same value.

It remains to prove Lemma \ref{clm:exactReduction}.

\begin{proof}[Proof of Lemma \ref{clm:exactReduction}]
Fix $\xxa$, and let $\xxaux$ be arbitrary.
For each $i \in [m]$ and $j \in S_i$, the set of the auxiliary equations created for $\AA_i$, we denote
\[
	\eps_i \defeq \AA_i \xxa - \cca_i, ~ 
	\hat{\delta}_i \defeq \AAhat_i \xxb - \cca_i, ~ 
	\delta_j \defeq \BBhat_j \xxb.
\]
Then 
\[
\eps_i =	\hat{\delta}_i + \sum_{j \in S_i} \delta_j 
\]

By the Cauchy-Schwarz inequality and $w_i = \alpha m_i$,
\[
\eps_i^2 = \left( \hat{\delta}_i + \sum_{j \in S_i} \delta_j \right)^2
\le \left( \hat{\delta}_i^2 + w_i \sum_{j \in S_i} \delta_j^2 \right) \left( 1 + \frac{m_i}{w_i} \right)
=  \left( \hat{\delta}_i^2 + w_i \sum_{j \in S_i} \delta_j^2 \right) \cdot \frac{\alpha+1}{\alpha}
\label{eqn:cauchy}
\]
that is, 
\[
	(\AA_i \xx^A - \cc_i^A)^2 \le \frac{\alpha +1}{\alpha} 
	\norm{\begin{bmatrix}
			\AAhat_i \\
			\WW^{1/2}_i \BBhat_i
		\end{bmatrix}  \xxb
		-  \begin{bmatrix}
			\cca_i \\
			{\bf 0}
		\end{bmatrix}
		}_2^2.
\]
The equality holds if and only if
\begin{align}
\label{eqn:cs_condition}
\hat{\delta}_i = w_i \delta_j,\qquad \forall j \in S_i.
\end{align}
Sum over all the rows $i \in [n_A]$: 
\begin{align}
\norm{\AA\xxa - \cca}_2^2
\leq 
\frac{\alpha+1}{\alpha}
\norm{\BB\begin{bmatrix}
		\xxa \\
		\xxaux
	\end{bmatrix} - \ccb
}_2.
\label{eqn:QFineq}
\end{align}

It remains to show 
when we take the minimum over $\xxaux$, the right-hand side of Eq. \eqref{eqn:QFineq} equals the left-hand side. 
That is, for any fixed $\xxa$, there exists $\xxaux$ such that Eq. \eqref{eqn:cs_condition}  holds.

In particular, we will momentarily prove the following Claim.
\begin{claim}
\label{clm:errorsObtainable}
 For any fixed $\xxa$ and its associated error
$\eps_i$ for each row $i$ of $\AA$ such that $\AA_i$ is neither 
a difference equation nor an average equation, the following linear system 
\begin{align}
\AAhat_i \begin{bmatrix}
	\xxa \\
	\xxaux
\end{bmatrix}
 & = \cca_i + \frac{\alpha}{\alpha+1} \eps_i,
\label{eqn:error_ls_a} \\
\BBhat_j \begin{bmatrix}
	\xxa  \\
	\xxaux
\end{bmatrix}
 & = \frac{1}{(\alpha+1)m_i} \eps_i, \forall j \in S_i.
\label{eqn:error_ls_b}
\end{align} 
has a solution (which may not be unique).
\end{claim}
Since every auxiliary variable is associated with only one row $i$ of
$\AA$, Claim~\ref{clm:errorsObtainable} implies that we can choose $\xxaux$ s.t. all these
linear subsystems are satisfied simultaneously.
Given such a choice of $\xxaux$, 
Eq. ~\eqref{eqn:cs_condition} is satisfied and thus
\[
\norm{\AA\xxa - \cca}_2^2
=
\frac{\alpha+1}{\alpha}
\norm{\BB\begin{bmatrix}
		\xxa \\
		\xxaux
	\end{bmatrix}- \ccb
}_2^2
.
\]
Together with Eq. \eqref{eqn:QFineq},
this completes the proof of Lemma \ref{clm:exactReduction}.
\end{proof}

\begin{proof}[Proof of Claim~\ref{clm:errorsObtainable}]

We will construct an assignment to all the variables of
$\xxaux$ such that Eq. \eqref{eqn:error_ls_a}
and~\eqref{eqn:error_ls_b} are satisfied.
We start with an assignment $\xxa$ to the main variables, and we
then assign values to auxiliary variables in the order that they are
created by the algorithm~{\algGZTtoMCT}.
Note that we will refer to variables $j_k$ and $l_k$ only in the
context of a fixed value of $t$, which always ensures that they are
unambiguously defined.
When the algorithm processes pair $\xxb_{j_k},\xxb_{l_k} = \uu_{j_k},\uu_{l_k}  $,
the value of these variables will have been set already,  but
$\xxb_{t} = \uu_t$ and the other newly created auxiliary variables have not.
Suppose the new equation we added to $\mathcal{B}_i$ is 
$s\cdot 2^{r} 
\setof{ \left( 2 \uu_{t} - \left( \uu_{j_{k}} + \uu_{l_{k}}
\right)\right) = 0}$
in line~\ref{lin:updateBi}.
We simply assign $\uu_t$ such that 
\[
	s\cdot 2^{r} 
	 \left( 2 \uu_{t} - \left( \uu_{j_{k}} + \uu_{l_{k}}
	\right)\right) = \frac{1}{(\alpha+1) m_i} \eps_i.
\]
Every auxiliary variable is associated with only one row, so we never
get multiple assignments to a variable using this procedure.

By the above setting, Eq. \eqref{eqn:error_ls_b} is satisfied.
It remains to check Eq. \eqref{eqn:error_ls_a}.
Sum up all the above equations over the auxiliary equations in $\mathcal{B}_i$ with minus sign
and the original equation $\AA_i \xxa = \cca_i + \eps_i$:
\begin{align*}
	\AAhat_i \xx = \cca_i + \eps_i - \frac{1}{\alpha+1} \eps_i
	= \cca_i +  \frac{\alpha}{\alpha+1} \eps_i.
\end{align*}
This completes the proof.
\end{proof}

\subsubsection{Relation to Schur Complements}
\label{subsec:SCRelation}

Note that we can write $\BB$ as $(\BB^{\textnormal{A}} ~~ \BB^{\textnormal{aux}})$, where $\BB^{\textnormal{A}}$ corresponds to 
the original variables and $\BB^{\textnormal{aux}}$ the auxiliary variables.
Then, 
\begin{align*}
	\BB^\top \BB = 
	\begin{bmatrix}
		\begin{array}{cc}
			(\BB^{\textnormal{A}})^\top \BB^{\textnormal{A}} & (\BB^{\textnormal{A}})^\top \BB^{\textnormal{aux}} \\
			(\BB^{\textnormal{aux}})^\top \BB^{\textnormal{A}} & (\BB^{\textnormal{aux}})^\top \BB^{\textnormal{aux}}
		\end{array}
	\end{bmatrix}.
\end{align*}
Lemma \ref{clm:exactReduction} essentially states that $\frac{\alpha}{\alpha+ 1} \AA^\top \AA$
is the Schur complement of $(\BB^{\textnormal{aux}})^\top \BB^{\textnormal{aux}}$ of $\BB \BB^\top$.
See the following definition and a fact of Schur complement.

\begin{definition}[Schur complement]
Let $\CC \in \mathbb{R}^{n \times n}$ be a $2 \times 2$ block matrix: 
$\CC = \left(\begin{array}{cc}
\CC_{11} & \CC_{12} \\
\CC_{12}^{\trp} & \CC_{22}
\end{array} \right)$ 
The \emph{Schur complement} of the block $\CC_{22}$ of $\CC$ is
\[
\CC \slash \CC_{22} \defeq \CC_{11} - \CC_{12}\CC^{-1}_{22}\CC_{12}^{\trp}.
\]
If $\CC_{22}$ is not invertible, then we replace the inverse with the pseudo-inverse
\end{definition}

Schur complement arises from block Gaussian elimination.
Schur complement has the following important fact.

\begin{fact}[Schur complement and minimizer]
For any fixed vector $\xx$, 
\[
\min_{\yy} \left( \begin{array}{cc}
\xx^{\trp} & \yy^{\trp}
\end{array} \right) \CC \left( \begin{array}{c}
\xx \\
\yy
\end{array} \right) 
= \xx^{\trp} (\CC \slash \CC_{22}) \xx.
\]
\label{clm:schur}
\end{fact}
\begin{proof}
We expand the left hand side,
\begin{align}
\left( \begin{array}{cc}
\xx^{\trp} & \yy^{\trp}
\end{array} \right) \CC \left( \begin{array}{c}
\xx \\
\yy
\end{array} \right) 
= \xx^{\trp}\CC_{11}\xx
+ 2\xx^{\trp} \CC_{12}\yy + \yy^{\trp} \CC_{22} \yy.
\label{eqn:schur_lhs}
\end{align}
Taking derivative w.r.t. $\yy$ and setting it to be 0 give that
\[
2\CC_{22}\yy + 2 \CC_{12}^{\trp} \xx = {\bf 0}.
\]
Plugging $\yy = - \CC_{22}^{\dagger} \CC_{12}^{\trp} \xx$ into Eq. \eqref{eqn:schur_lhs}, 
\[
\min_{\yy} \left( \begin{array}{cc}
\xx^{\trp} & \yy^{\trp}
\end{array} \right) \CC \left( \begin{array}{c}
\xx \\
\yy
\end{array} \right) 
= \xx^{\trp}\CC_{11}\xx - \xx^{\trp} \CC_{12}\CC_{22}^{\dagger} \CC_{12}^{\trp} \xx
= \xx^{\trp} (\CC \slash \CC_{22}) \xx.
\]
This completes the proof.
\end{proof}


\newcommand{\zzn}{\tilde{\zz}}
\newcommand{\zzp}{\hat{\zz}}

\subsection{Relation Between Approximate Solvers}
\label{subsec:GZ2ToMC2apx}

We now show that approximate solvers for $\BB$ also
translate to approximate solvers for $\AA$.

\begin{lemma}
\label{lem:mctogenerr}
Let \lsa$(\AA,\cca, \epsa)$ be an instance over $\calG_{z,2}$ where $\epsa \in (0,1)$. 
Let $(\BB,\ccb) \assign \algGZTtoMCT(\AA, \cca, \alpha = 1)$, and 
\[
\eps^{\textnormal{B}} \le  \frac{\epsa}{\sqrt{n_A m_A} \cdot  \norm{\AA}_{\max} \norm{\cca}_2}.
\]
Suppose $\xxb$ is a solution to \lsa$(\BB,\ccb,\eps^{\textnormal{B}})$, and 
$\xxa \leftarrow \algMCTtoGZTsoln(\AA, \BB, \cca, \xxb)$. 
Then, $\xxa$ is a solution to \lsa$(\AA,\cca,\epsa)$.
\end{lemma}

\begin{proof}

  Write $\xxb = \begin{bmatrix}
    \xxa \\
    \xxaux
  \end{bmatrix}$.
  By Lemma \ref{clm:exactReduction},
  \begin{align*}
    & \norm{\AA\xxa - \cca}_2^2 \le \frac{\alpha}{\alpha + 1}  \norm{\BB \xxb - \ccb}_2^2 \\
    & \norm{(\II - \PPi_{\AA}) \cca}_2^2  =   \frac{\alpha}{\alpha + 1} \norm{(\II - \PPi_{\BB}) \ccb}_2^2
  \end{align*}
  where $\alpha = 1$ in our setting.
  Then,
  \begin{align*}
    \norm{\AA\xxa - \PPi_{\AA} \cca}_2^2 
    & \le \frac{\alpha}{\alpha+1} \norm{\BB \xxb - \PPi_{\BB} \ccb}_2^2  \\
    & \le \frac{\alpha}{\alpha+1} (\eps^{\textnormal{B}})^2 \norm{\PPi_{\BB}\ccb}_2^2
    \\
    & \le \frac{\alpha}{\alpha+1} (\eps^{\textnormal{B}})^2 \norm{\ccb}_2^2 \lambda_{\max}(\AA^\top \AA) \norm{\PPi_{\AA} \cca}_2^2 
    \tag*{by Lemma \ref{lem:regressions_approx}}\\
    & \le (\epsa)^2 \norm{\PPi_{\AA} \cca}_2^2 
    \tag*{since $\lambda_{\max}(\AA^\top \AA) \le n_A m_A \norm{\AA}_{\max}^2$}
  \end{align*}
\end{proof}

\subsection{The Condition Number of the New Matrix}
\label{subsec:GZToMC2ConditionNumber}

In this section, we show that the condition number of $\BB$ is 
upper bounded by the condition number of $\AA$ up to a $\poly(n)$ multiplicative factor.

\begin{lemma}
  Let $(\BB,\ccb) \leftarrow \algGZTtoMCT(\AA, \cca, \alpha = 1)$.
  Then, $\kappa(\BB) = O \left(
    \max \left\{ \frac{\nnz(\AA)^{7/2} \norm{\AA}_{\max} \log^2  \norm{\AA}_{\max}}{\lambda_{\min} (\AA^\top \AA)},
    ~ \nnz(\AA)
      \right\}
  \right)$.
\end{lemma}

Note that $\lambda_{\min}(\AA^\top \AA)$ is polynomially bounded, since both $\kappa(\AA^\top \AA)$ 
and $\lambda_{\max}(\AA^\top \AA)$ are polynomially bounded.

\subsubsection{The Maximum Eigenvalue}

We start with a simple observation.
\begin{claim}
  $\max_{i} \WW_{i,i} = O(\nnz(\AA) \log \norm{\AA}_{\max})$.
\end{claim}

We bound the maximum eigenvalue of the constructed matrix $\BB$.

\begin{lemma}
  Let $\BB \in \mathbb{R}^{m_B \times n_B}$ be returned by $\algGZTtoMCT(\AA,\cca, \alpha = 1)$.
  Then, $\lambda_{\max}(\BB^\top \BB) = O( \nnz(\AA)^2 \norm{\AA}_{\max} \log \norm{\AA}_{\max})$.
  \label{lem:max_eig_bb_appendix}
\end{lemma}

\begin{proof}

  Write $\BB = \WWtil^{1/2} \BBtil$. 
  By the Courant-Fischer theorem,
  \begin{align*}
    \lambda_{\max}(\BB^\top \BB) 
    & = \max_{\xx: \norm{\xx}_2 = 1} \xx^\top \BBtil^\top \WWtil \BBtil \xx 
     = \max_{\xx: \norm{\xx}_2 = 1} \sum_{i=1}^{m_B} \WWtil_{i,i} (\BBtil_i \xx)^2. 
  \end{align*} 
  Note that 
  \begin{align*}
   & \WWtil_{i,i} = O(\nnz(\AA) \log \norm{\AA}_{\max}), \\
   &  (\BBtil_i \xx)^2 \le 6 \norm{\AA}_{\max} \sum_{j=1}^{n_B} \xx_j^2 \cdot \abs{\{i \in [m_B]: \BB_{ij} \neq 0 \}}. 
  \end{align*}
  Thus,
  \[
    \lambda_{\max}(\BB^\top \BB) =  O( \nnz(\AA)^2 \norm{\AA}_{\max} \log \norm{\AA}_{\max}).
  \]
\end{proof}

\subsubsection{The Minimum Nonzero Eigenvalue}

To apply the Courant-Fischer theorem for the minimum nonzero eigenvalue of $\BB$, 
we need to first characterize the null space of $\BB$.
Given any $\xxa \in \Null(\AA)$ with dimensions $n_A$, we extend $\xxa$ to a vector in dimension $\nb$:
\begin{align}
\ppxa \defeq \begin{bmatrix}
	\xxa \\
	\xxaux
\end{bmatrix}.
\label{eqn:GZ2toMC2apx_pxa_def}
\end{align}
where we assign the values of the auxiliary variables $\xxaux$ in the order that they are created in Algorithm~\ref{alg:GZ2toMC2}. 
In an auxiliary equation created in line~\ref{lin:updateBi}, 
$\uu_{j_k}$ and $\uu_{l_k}$ have already been assigned, and we simply assign 
$\uu_t$ to such that the new equation holds.

\begin{lemma}
\label{lem:GZtoMC2Null}
$\Null(\BB) = \Span\{ \ppxa: \xxa \in \Null(\AA) \}$.
\end{lemma}
\begin{proof}
Let $S = \Span\{ \ppxa: \xxa \in \Null(\AA) \}$. 
We can check that $\Null(\BB) \supseteq S$.
Then, for any $\xx \in \Null(\BB)$,
by Lemma~\ref{clm:exactReduction} with $\cca = {\bf 0}$, we have
$\AA\xxa = {\bf 0}$, that is, $\xxa \in \Null(\AA)$.
Thus, $\Null(\BB) \subseteq S$.
\end{proof}

\begin{lemma}
  $\lambda_{\min}(\BB^\top \BB) 
  = \Omega \left(
    \min \left\{ \frac{\lambda_{\min}(\AA^\top \AA)}{ \nnz(\AA)^{3/2} \log \norm{\AA}_{\max}}, 
    ~  \nnz(\AA) \norm{\AA}_{\max} \log \norm{\AA}_{\max}  \right\}    
  \right)$.
\end{lemma}

\begin{proof}
Let $\xx = \begin{bmatrix}
  \xxa \\
  \xxaux
\end{bmatrix} \in \mathbb{R}^{n_B}$ be an arbitrary unit vector orthogonal to $\Null(\BB)$,
let $\yy = \xx - \pp(\xx^A)$, and let 
\[
\delta \defeq \min \left\{ \frac{\lambda_{\min}(\AA^\top \AA)}{C \nnz(\AA)^{5/2} \norm{\AA}_{\max} \log \norm{\AA}_{\max}}, 
~\frac{1}{5}
  \right\}
 .
\]
where $C$ is a constant such that $\lambda_{\max}(\BB^\top \BB) \le C \nnz(\AA)^2 \norm{\AA}_{\max} \log \norm{\AA}_{\max}$
(by Lemma \ref{lem:max_eig_bb_appendix}, such $C$ exists).

We will prove the statement by exhausting the cases of $\norm{\yy}_{\infty}$.

Suppose $\norm{\yy}_{\infty} > \delta$. 
Then, there must exist an auxiliary equation 
 $2^r(2\uu_t - (\uu_{j_k} + \uu_{l_k})) = 0$ such that 
$\abs{2\yy_t - (\yy_{j_k} + \yy_{l_k})} > \frac{\delta}{\log \norm{\AA}_{\max}}$, 
otherwise $\norm{\yy}_{\infty} < \delta$. 
Then, 
\begin{align*}
  \xx^\top \BB^\top \BB \xx \ge \nnz(\AA) \cdot \norm{\AA}_{\max} \cdot \frac{\delta}{\log \norm{\AA}_{\max}}. 
\end{align*}

Suppose $\norm{\yy}_{\infty} \le \delta$. Then, 
\begin{align*}
  & \norm{\pp(\xxa)}_2^2 \ge \norm{\xx}_2^2 - 2 \norm{\xx}_2 \norm{\yy}_2
  \ge 1 - 2 \delta,  \\
  & \norm{\pp(\xxa)}_2^2 \le \norm{\xx}_2^2 + \norm{\yy}_2^2 \le 1 + \delta^2 n_B.
\end{align*}
Then, 
\begin{align*}
  \xx^\top \BB^\top \BB \xx 
  & = (\pp(\xx^A) + \yy)^\top \BB^\top \BB (\pp(\xx^A) + \yy) \\
  & \ge  \pp(\xx^A)^\top \BB^\top \BB \pp(\xx^A)
   - 2 \norm{\BB \pp(\xx^A)}_2 \norm{\BB \yy}_2 \\
  & \ge \lambda_{\min}(\AA^\top \AA) \norm{\pp(\xx^A)}_2^2
- 2\lambda_{\max}(\BB^\top \BB) \norm{\pp(\xx^A)}_2 \norm{\yy}_2 \\
& \ge \lambda_{\min}(\AA^\top \AA) (1 - 2\delta)
- \lambda_{\max}(\BB^\top \BB)  \delta \sqrt{n_B} \sqrt{1 + \delta^2 n_B} \\
& \ge \frac{1}{2} \lambda_{\min}(\AA^\top \AA)
\end{align*}
where the last inequality is by the choice of $\delta$.
By the Courant-Fischer theorem, 
\[
\lambda_{\min}(\BB \BB^\top) \ge 
\min \left\{ \frac{\lambda_{\min}(\AA^\top \AA)}{C \nnz(\AA)^{3/2} \log \norm{\AA}_{\max}}, 
~ \frac{1}{5} \nnz(\AA) \norm{\AA}_{\max} \log \norm{\AA}_{\max}  \right\}. 
\]
\end{proof}



\section{Reducing Solving 2-Complex Boundary Linear Equations to Combinatorial Laplacian Linear Equations}
\label{sect:appendix_lap2boundary}

In this section, we formally state Theorem \ref{thm:compLap2Boundary_informal} as below
and provide a proof.
Recall that we use 
$\sigma_{\min}(\AA)$ to denote the smallest non-zero eigenvalue.

\begin{theorem}
	Let $\calL_1 = \partial_1^\top \partial_1 + \partial_2 \partial_2^\top \in \mathbb{R}^{m \times m}$ be the combinatorial Laplacian of a 2-complex.
	Let $\dd \in \mathbb{Z}^m$. Suppose we can solve \lsa~ $(\calL_1, \dd, \eps)$ in time $\Otil(\nnz(\calL_1)^c)$ 
	where $c \ge 1$ is a constant. Then, we can solve \lsa~ $(\partial_2, \dd, \delta)$ in time 
	$\Otil(\nnz(\partial_2)^c)$ by choosing \[\eps < \delta
    \frac{\sigma_{\min}(\calL_1)^{1/2}}
    {\sigma_{\max}(\partial_2 )^2}
    \frac{1}
    {\norm{\dd}_2}.\]
	\label{thm:compLap2Boundary}
\end{theorem}

\begin{proof}

Suppose $\xx_1$ satisfies
\begin{align*}
	& \norm{\calL_1 \xx_1 - \PPi_{\calL_1} \dd}_2 \le \eps \norm{\PPi_{\calL_1} \dd}_2 \\
\end{align*}
By our assumption, we can compute $\xx_1$ in time $\Otil (\nnz(\calL_1)^c)$.
We choose
\[
  \ff = \partial_2^\top \xx_1
\]
We claim that $\ff$ solves \lsa~ $(\partial_2, \dd, \delta)$.
Since $\partial_1 \partial_2 = \bf{0}$, we have 
$
\calL_1^{\dagger} = (\partial_1^\top \partial_1)^{\dagger} + (\partial_2 \partial_2^\top)^{\dagger}
$
and
$
\PPi_{\partial_2}  (\partial_1^\top \partial_1)^{\dagger} \dd = \veczero.
$
Then, 
\begin{align*}
\norm{\partial_2 \ff - \PPi_{\partial_2} \dd}_2
  & = \norm{\xx_1 - (\partial_2 \partial_2^\top )^{\dagger} \dd}_{(\partial_2\partial_2^\top)^2}
  \\
  & \leq \sigma_{\max}(\partial_2 ) \norm{\xx_1
    - (\partial_2 \partial_2^\top )^{\dagger}
    \dd}_{\partial_2\partial_2^\top}
      \\
  & = \sigma_{\max}(\partial_2 ) \norm{\PPi_{\partial_2} (\xx_1 
    - (\partial_2 \partial_2^\top )^{\dagger} \dd)}_{\partial_2\partial_2^\top}
      \\
  & = \sigma_{\max}(\partial_2 ) \norm{\PPi_{\partial_2} (\xx_1 
    - \calL_1^{\dagger} \dd)}_{\partial_2\partial_2^\top}
      \\
  & = \sigma_{\max}(\partial_2 ) \norm{\xx_1 
    - \calL_1^{\dagger} \dd}_{\partial_2\partial_2^\top}
  \\
  & \leq \sigma_{\max}(\partial_2 ) \norm{\xx_1 
    - \calL_1^{\dagger} \dd}_{\calL_1}
    \\
  & \leq \frac{\sigma_{\max}(\partial_2 )}
    {\sigma_{\min}(\calL_1)^{1/2}}
    \norm{\calL_1 \xx_1 - \PPi_{\calL_1} \dd}_2
    \\
  & \leq \frac{\sigma_{\max}(\partial_2 )}
    {\sigma_{\min}(\calL_1)^{1/2}}
    \eps \norm{\PPi_{\calL_1} \dd}_2
  \\
  & \leq \frac{\delta}
    {\sigma_{\max}(\partial_2 )}
    \\
  & \leq
    \delta
    \norm{\PPi_{\partial_2} \dd}_2 
    \tag*{by Claim \ref{lm:projection_norm}} 
\end{align*}
\end{proof}


\section{Connections With Interior Point Methods}
\label{sect:appendix_IPM}
We now show that in order to solve a generalized maxflow problem in a 2-complex flow network using an Interior Point Method (IPM), it suffices to be able to apply the pseudo-inverse of $\partial_2 \WW \partial_2^\top$ for diagonal positive weight matrices $\WW$ (and this problem is essentially equivalent to applying the pseudo-inverse of the combinatorial Laplacian of the complex, c.f. Section \ref{sect:boundary_and_combinatorial}).
We sketch how the these pseudo-inverse problems arise when solving a generalized maxflow using IPM, which is motivated by \cite{Mad16}. For the more curious readers, we recommend the book \cite{Ren01} for a complete view of general IPM algorithms.

Given a 2-complex flow network $\mathcal{K}$ with $m$ edges and $t$ triangles, a non-negative capacity vector $\cc\in\R^t_{\geq 0}$, and a demand vector $\gamma\in\R^m$ such that $\gamma\in\im(\partial_2)$. The $\gamma$-maxflow problem is formulated by the following linear programming:
\begin{equation}
	\begin{aligned}
		\label{eq:LP_maxflow_}
		\max_{F,\ff} &~~ F\\
		\text{s.t.} &~~ \partial_2\ff=F\gamma\\
		&~~ -\cc\leq \ff\leq \cc
	\end{aligned}
\end{equation}
The $\gamma$-maxflow in 2-complex flow networks is a generalization of $s$-$t$ maxflow in graphs. The first constraint encodes the conservation of flows for edges in $\mathcal{K}$. And the second constraint forces the flow on triangles to satisfy the capacity constraints.

We call $F$ the flow value of $\ff$ when $\partial_2\ff=F\gamma$. We assume that the optimal flow value $F^*$ is known by IPM algorithms, which can be estimated by the binary search.

The main idea of IPM is to get rid of inequality constraints by using barrier functions, and then apply Newton's method to a sequence of equality constrained problems.
The most widely used barrier function is logarithmic barrier function, which in the $\gamma$-maxflow problem gives
\[V(\ff)=\sum_{\Delta\in [t]}-\log(\cc(\Delta)-\ff(\Delta))-\log(\cc(\Delta)+\ff(\Delta)).\]
Then for a given $0\leq \alpha<1$, we define the following Barrier Problem:
\begin{equation}
	\begin{aligned}
		\label{eq:LP_maxflow_barrier}
		\min_{\ff} &~~ V(\ff)\\
		\text{s.t.} &~~ \partial_2\ff=\alpha F^*\gamma
	\end{aligned}
\end{equation}
We start with zero flow, i.e., $\alpha_0=0$, and then increase $\alpha_{i+1}=\alpha_i+\alpha'$ gradually in each iteration to make progress. 
Given a small enough $\alpha'$, each iteration is composed of a \textit{progress step} and a \textit{centering step}. 

\paragraph{Progress step}
we first take a progress step by making a Newton step to Problem \eqref{eq:LP_maxflow_barrier} at the current point $\ff$, while increasing the flow value by $\alpha'$, which gives
\begin{equation}
	\begin{aligned}
		\label{eq:LP_maxflow_barrier_newton1}
		\min_{\boldsymbol{\delta}} &~~ \gg^\top(\ff)\boldsymbol{\delta}+\frac{1}{2}\boldsymbol{\delta}^\top\HH(\ff)\boldsymbol{\delta}\\\
		\text{s.t.} &~~ 
		\partial_2\boldsymbol{\delta}=\alpha' F^*\gamma
	\end{aligned}
\end{equation}
where $\gg(\ff)$ and $\HH(\ff)$ are the gradient and Hessian of $V$ at the current point $\ff$, respectively.

Problem \eqref{eq:LP_maxflow_barrier_newton1} has the Lagrangian
\[\mathcal{L}(\boldsymbol{\delta},\xx)=\gg^\top(\ff)\boldsymbol{\delta}+\frac{1}{2}\boldsymbol{\delta}^\top\HH(\ff)\boldsymbol{\delta}+\xx^\top(\alpha' F^*\gamma-\partial_2\boldsymbol{\delta}).\]
Using optimality condition, we have
\[\nabla_{\boldsymbol{\delta}}\mathcal{L}(\boldsymbol{\delta},\xx)=\gg(\ff)+\HH(\ff)\boldsymbol{\delta}-\partial_2^\top\xx=\mathbf{0},\]
which gives,
\begin{equation*}
	\label{eq:delta}
	\boldsymbol{\delta}=\HH^{-1}(\ff)(\partial_2^\top\xx-\gg(\ff))
\end{equation*}
Multiplying $\partial_2$ in both sides and using the constraint $\partial_2\boldsymbol{\delta}=\alpha' F^*\gamma$, we obtain
\begin{equation*}
	\label{eq:x}
	\partial_2\HH^{-1}(\ff)\partial_2^\top\xx=\partial_2\HH^{-1}(\ff)\gg(\ff)+\alpha' F^*\gamma.
\end{equation*}
Thus, we have shown that it suffices to apply the pseudo-inverse of $\partial_2\HH^{-1}(\ff)\partial_2^\top$ to solve $\xx$ and $\boldsymbol{\delta}$:
\[\xx=\left(\partial_2\HH^{-1}(\ff)\partial_2^\top\right)^\dagger(\partial_2\HH^{-1}(\ff)\gg(\ff)+\alpha' F^*\gamma),\]
\[\boldsymbol{\delta}=\HH^{-1}(\ff)\partial_2^\top\left(\partial_2\HH^{-1}(\ff)\partial_2^\top\right)^\dagger\left(\partial_2\HH^{-1}(\ff)\gg(\ff)+\alpha' F^*\gamma\right)-\HH^{-1}(\ff)\gg(\ff).\]

\paragraph{Centering step}
We then take a centering step by making a Newton step to Problem \eqref{eq:LP_maxflow_barrier} at the updated point of $\tilde{\ff}\defeq \ff+\boldsymbol{\delta}$ without increasing the flow value, which gives
\begin{equation}
	\begin{aligned}
		\label{eq:LP_maxflow_barrier_newton2}
		\min_{\tilde{\boldsymbol{\delta}}} &~~ \gg^\top(\tilde{\ff})\tilde{\boldsymbol{\delta}}+\frac{1}{2}\tilde{\boldsymbol{\delta}}^\top\HH(\tilde{\ff})\tilde{\boldsymbol{\delta}}\\\
		\text{s.t.} &~~ \partial_2\tilde{\boldsymbol{\delta}}=\mathbf{0}
	\end{aligned}
\end{equation}
Similar to the progress step, it suffices to apply the pseudo-inverse of $\partial_2\HH^{-1}(\fftil)\partial_2^\top$ to solve $\tilde{\boldsymbol{\delta}}$:
\[\tilde{\boldsymbol{\delta}}=\left(\HH^{-1}(\fftil)\partial_2^\top\left(\partial_2\HH^{-1}(\fftil)\partial_2^\top\right)^\dagger\partial_2-\II\right)\HH^{-1}(\fftil)\gg(\fftil).\]


\end{document}